\documentclass[a4paper,runningheads,USenglish]{llncs}

\usepackage{subfig}
\usepackage{microtype}
\usepackage{amsmath}
\usepackage{amssymb}
\usepackage{thm-restate}
\usepackage{dsfont}
\usepackage{xspace}
\usepackage{color}

\definecolor{blue}{rgb}{0.274,0.392,0.666}
\definecolor{darkblue}{rgb}{0.063,0.306,0.545}
\definecolor{red}{rgb}{1,0.3,0.3}
\definecolor{greennn}{rgb}{0,0.588,0.509}

\newcommand{\red}[1]{{{\textcolor{black}{#1}\xspace}}}

\newcommand{\blue}[1]{{{\textcolor{darkblue}{#1}\xspace}}}

\usepackage{hyperref}
\hypersetup{colorlinks,linkcolor={darkblue},citecolor={darkblue},urlcolor={darkblue}} 
\usepackage{paralist}
\usepackage[color=yellow]{todonotes}
\usepackage{rotating}
\usepackage{algorithm}
\usepackage[noend]{algpseudocode}
\usepackage{multirow}
\usepackage{dsfont}
\usepackage{rotating}
\usepackage{stmaryrd}
\usepackage{rotating}
\usepackage{lineno}
\usepackage{framed}

\usepackage{ntheorem}

\renewtheorem{claim}{Claim}[theorem]

\let\doendproof\endproof
\renewcommand\endproof{~\hfill\qed\doendproof}

\newenvironment{claimproof}{\noindent{\em Proof of the claim.}}{\hspace*{\fill}\qed\vspace{2mm}}

\renewcommand{\paragraph}[1]{\smallskip\noindent\textbf{#1}\xspace}

\renewcommand{\paragraph}[1]{\smallskip\noindent\textbf{#1}\xspace}

\newcommand{\remove}[1]{}

\newif\ifshort
\shorttrue
\shortfalse

\graphicspath{{figures/}}

\begin{document}

\authorrunning{G. Da Lozzo et al.}

\title{
%Algorithms to Morph \\upward planar Graph Drawings 
%On Morphing upward planar Graph Drawings 
Upward Planar Morphs\thanks{This research was partially supported by MIUR Project ``MODE'', by H2020-MSCA-RISE project ``CONNECT'', and by MIUR-DAAD JMP N$^\circ$\ 34120.\xspace}}
\author{{Giordano Da Lozzo}, {Giuseppe Di Battista}, {Fabrizio Frati}, \\{Maurizio Patrignani}, and {Vincenzo Roselli}}
\institute{
{
Roma Tre University, Rome, Italy\\
\{\href{mailto:dalozzo@dia.uniroma3.it}{dalozzo},\href{mailto:gdb@dia.uniroma3.it}{gdb},\href{mailto:frati@dia.uniroma3.it}{frati},\href{mailto:patrigna@dia.uniroma3.it}{patrigna},\href{mailto:roselli@dia.uniroma3.it}{roselli}\}\href{mailto:dalozzo@dia.uniroma3.it,gdb@dia.uniroma3.it,frati@dia.uniroma3.it,patrigna@dia.uniroma3.it,roselli@dia.uniroma3.it}{@dia.uniroma3.it}}
}
\maketitle

\begin{abstract}
We prove that, given two topologically-equivalent upward planar straight-line
drawings of an $n$-vertex directed graph $G$, there always exists a morph
between them such that all the intermediate drawings of the morph are upward planar and
straight-line. Such a morph consists of $O(1)$ morphing steps if $G$ is a
reduced planar $st$-graph, $O(n)$ morphing steps if $G$ is a planar $st$-graph,
$O(n)$ morphing steps if $G$ is a reduced upward planar graph, and $O(n^2)$
morphing steps if $G$ is a general upward planar graph.
Further, we show that $\Omega(n)$ morphing steps might be necessary for an upward planar morph between two topologically-equivalent upward planar straight-line
drawings of an $n$-vertex path.

\end{abstract}

%%%%%%%%%%%%%%%%%%%%%%%%%%%%%%%%%%%%%%%%%%%%%%%%%%%%%%%%%%%%%%%%%%%%%%
%%%%%%%%%%%%%%%%%%%%%%%%%%%%%%%%%%%%%%%%%%%%%%%%%%%%%%%%%%%%%%%%%%%%%%
%%%%%%%%%%%%%%%%%%%%%%%%%%%%%%%%%%%%%%%%%%%%%%%%%%%%%%%%%%%%%%%%%%%%%%
%	    ____      __                 __           __  _
%	   /  _/___  / /__________  ____/ /_  _______/ /_(_)___  ____
%	   / // __ \/ __/ ___/ __ \/ __  / / / / ___/ __/ / __ \/ __ \
%	 _/ // / / / /_/ /  / /_/ / /_/ / /_/ / /__/ /_/ / /_/ / / / /
%	/___/_/ /_/\__/_/   \____/\__,_/\__,_/\___/\__/_/\____/_/ /_/
%%%%%%%%%%%%%%%%%%%%%%%%%%%%%%%%%%%%%%%%%%%%%%%%%%%%%%%%%%%%%%%%%%%%%%
%%%%%%%%%%%%%%%%%%%%%%%%%%%%%%%%%%%%%%%%%%%%%%%%%%%%%%%%%%%%%%%%%%%%%%
%%%%%%%%%%%%%%%%%%%%%%%%%%%%%%%%%%%%%%%%%%%%%%%%%%%%%%%%%%%%%%%%%%%%%%

\section{Introduction}\label{se:intro}

One of the definitions of the word \emph{morph} that can be found in English
dictionaries is ``to gradually change into a different image''. The
Graph Drawing community defines the morph of graph drawings similarly. Namely, given two drawings $\Gamma_0$ and
$\Gamma_1$ of a graph $G$, a \emph{morph} between $\Gamma_0$ and
$\Gamma_1$ is a continuously changing family of drawings of $G$ indexed by time
$t \in [ 0,1 ]$, such that the drawing at time $t=0$ is
$\Gamma_0$ and the drawing at time $t=1$ is $\Gamma_1$.
Further, the way the Graph Drawing community adopted the word morph is
consistent with its Ancient Greek root $\mu \omega \rho \phi \acute\eta$, which
means ``shape'' in a broad sense. Namely, if both $\Gamma_0$ and
$\Gamma_1$ have a certain geometric property, it is desirable that all
the drawings of the morph also have the same property. In particular, we talk about a  \emph{planar},
a \emph{straight-line}, an \emph{orthogonal}, or a \emph{convex morph} if all the intermediate drawings of
the morph are \emph{planar} (edges do not cross), \emph{straight-line} (edges are straight-line segments), \emph{orthogonal} (edges are polygonal lines composed of horizontal and
vertical segments), or \emph{convex} (the drawings are planar and straight-line, and the faces are delimited by convex polygons), respectively.

The state of the art on planar morphs covers more than 100 years, starting from the 1914/1917 works of Tietze \cite{tietze-usaeqass-14} and Smith \cite{s-ocrsui-17}. The seminal papers of Cairns \cite{c-dprc-44} and Thomassen
\cite{t-dpg-83} proved the existence of a planar straight-line morph between any two topologically-equivalent planar  straight-line drawings of a graph. In the last 10 years, the attention of the research community focused on
algorithms for constructing planar morphs with few \emph{morphing steps} (see, e.g., \cite{aadddhlrssw-cpwlv-11,DBLP:journals/siamcomp/AlamdariABCLBFH17,aacdfl-mpgdpns-13-c,Angelini-optimal-14,DBLP:conf/compgeom/AngeliniLFLPR15,afpr-mpgde-13,Barrera-unidirectional,biedl2013morphing,BLS-orth,r-mvdg-14,DBLP:conf/compgeom/GoethemV18}).
Each morphing step, sometimes simply called {\em step}, is a \emph{linear} morph, in which the vertices move along straight-line (possibly distinct) trajectories at
uniform speed. A \emph{unidirectional} morph is a linear morph in which the vertex trajectories are all parallel.
It is known~\cite{DBLP:journals/siamcomp/AlamdariABCLBFH17,Angelini-optimal-14} that a planar straight-line morph with a linear number of unidirectional morphing steps exists between any two topologically-equivalent planar straight-line drawings of the same graph, and that this bound is the best possible.

%\cite{} for
%planar straight-line morphs, in \cite{} for
%convex morphs, and in \cite{} for
%planar orthogonal morphs. 

%
%, a subset of the vertices move ,
%along \emph{trajectories} of different types, to a new location. A \emph{linear}
%morphing step is such that vertices. 

\emph{Upward planarity} is usually regarded as the natural extension of planarity to directed graphs; see, e.g.,~\cite{DBLP:journals/algorithmica/BertolazziBLM94,DBLP:journals/siamcomp/BertolazziBMT98,DETT,DBLP:journals/tcs/BattistaT88,Garg:2002:CCU:586839.586865}. A drawing of a directed graph is
\emph{upward planar} if it is planar and the edges are represented by curves
monotonically increasing in the vertical direction. Despite the importance
of upward planarity, up to now, no algorithm has been devised
to morph upward planar drawings of directed graphs. This paper deals
with the following question: Given two topologically-equivalent upward planar drawings $\Gamma_0$ and $\Gamma_1$ of an upward planar directed graph $G$, does an \emph{upward planar straight-line morph} between $\Gamma_0$ and $\Gamma_1$ always exist? In this paper we give a positive answer to this question.

Problems related to upward planar graphs are usually more difficult than the corresponding problems for undirected graphs. For example, planarity can be tested in linear time \cite{Hopcroft:1974:EPT:321850.321852} while testing upward planarity is NP-complete \cite{Garg:2002:CCU:586839.586865}; all planar graphs admit planar straight-line grid drawings with polynomial area~\cite{Schnyder:1990:EPG:320176.320191} while there are upward planar graphs that require exponential area in any upward planar straight-line grid drawing \cite{DiBattista:1992:ARS:149153.149159}. Quite surprisingly, we show that, from the morphing point of view, the difference between planarity and upward planarity is less sharp; indeed, in some cases, upward planar straight-line drawings can be morphed even more efficiently than  planar straight-line drawings.

%The answer is not obvious since there are several problems that are tractable for planarity and that are difficult
%for upward planarity. E.g., 
%
More in detail, our results are as follows. Let $\Gamma_0$ and $\Gamma_1$ be topologically-equivalent upward planar drawings of an $n$-vertex upward plane graph $G$. We show algorithms to construct upward planar straight-line morphs between $\Gamma_0$ and $\Gamma_1$ with the following number of unidirectional morphing steps:
\begin{enumerate}[\bf i.]
\item $O(1)$ steps if $G$ is a reduced plane $st$-graph;
\item $O(n)$ steps if $G$ is a plane $st$-graph;
\item $O(n)$ steps if $G$ is a reduced upward plane graph;
\item $O(n\cdot f(n))$ steps if $G$ is a general upward plane graph, assuming that an $O(f(n))$-step algorithm exists to construct an upward planar morph between any two upward planar drawings of any $n$-vertex plane $st$-graph. This, together with Result \textbf{ii.}, yields an $O(n^2)$-step upward planar morph for general upward plane graphs.
\end{enumerate}
%%
%Results \textbf{i.} and \textbf{iii.} yield an $O(n)$-step morph for reduced
%upward planar graphs, while Results 
%%

Further, we show that there exist two topologically-equivalent upward planar drawings of an $n$-vertex upward plane path such that any upward planar morph between them consists of $\Omega(n)$ morphing steps.

In order to prove Result \textbf{i.} we devise a technique that allows us to construct a morph in which each morphing step modifies either only the $x$-coordinates or only the \mbox{$y$-coordinates} of the vertices. 
%
% To prove Result \textbf{ii.} we construct morphs of
% maximal planar $st$-graphs building on the techniques in
% \cite{DBLP:journals/siamcomp/AlamdariABCLBFH17}  and leveraging on the arrangements of
% low-degree vertices in upward drawings. Then, we exploit such morphs to morph general planar $st$-graphs.
%
Result \textbf{ii.} builds on the techniques in  \cite{DBLP:journals/siamcomp/AlamdariABCLBFH17} and leverages on the arrangement of low-degree vertices in upward planar drawings in order to morph maximal plane $st$-graphs. We then exploit such morphs for general plane $st$-graphs.
In order to prove Results \textbf{iii.} and \textbf{iv.} we use an inductive
technique for reducing the geometric differences
between $\Gamma_0$ and $\Gamma_1$.

The paper is organized as follows. In \autoref{se:preliminaries} we introduce preliminary definitions and notation. In \autoref{se:slow-morphs-fast-morphs} we prove a lower bound on the number of morphing steps that might be required by an upward planar morph and we present a technique for constructing upward planar morphs with few morphing steps. In
 \autoref{se:st-graphs} we study upward planar morphs of plane $st$-graphs. In  \autoref{se:general-graphs} we study upward planar morphs of general upward plane graphs. Finally, in \autoref{se:conclusions} we present conclusions and open problems.

%%%%%%%%%%%%%%%%%%%%%%%%%%%%%%%%%%%%%%%%%%%%%%%%%%%%%%%%%%%%%%%%%%%%%%
%%%%%%%%%%%%%%%%%%%%%%%%%%%%%%%%%%%%%%%%%%%%%%%%%%%%%%%%%%%%%%%%%%%%%%
%%%%%%%%%%%%%%%%%%%%%%%%%%%%%%%%%%%%%%%%%%%%%%%%%%%%%%%%%%%%%%%%%%%%%%
%	    ____            ___           _                  _
%	   / __ \________  / (_)___ ___  (_)___  ____ ______(_)__  _____
%	  / /_/ / ___/ _ \/ / / __ `__ \/ / __ \/ __ `/ ___/ / _ \/ ___/
%	 / ____/ /  /  __/ / / / / / / / / / / / /_/ / /  / /  __(__  )
%	/_/   /_/   \___/_/_/_/ /_/ /_/_/_/ /_/\__,_/_/  /_/\___/____/
%%%%%%%%%%%%%%%%%%%%%%%%%%%%%%%%%%%%%%%%%%%%%%%%%%%%%%%%%%%%%%%%%%%%%%
%%%%%%%%%%%%%%%%%%%%%%%%%%%%%%%%%%%%%%%%%%%%%%%%%%%%%%%%%%%%%%%%%%%%%%
%%%%%%%%%%%%%%%%%%%%%%%%%%%%%%%%%%%%%%%%%%%%%%%%%%%%%%%%%%%%%%%%%%%%%%

\section{Preliminaries}\label{se:preliminaries}

We assume familiarity with graph drawing~\cite{DETT} and related concepts. 

%\finalXXXXarxiv{The}{Anyway, for the reader's convenience, the} complete set of standard definitions \finalXXXXarxiv{is also given in the full version~\cite{arxiv}}{is given in the Appendix}.

\paragraph{Graph drawings.} In a drawing of a graph vertices are represented by distinct points of the plane and edges are represented by Jordan arcs connecting the points representing their end-vertices. In a \emph{straight-line drawing} the edges are represented by straight-line segments. In this paper we only consider straight-line drawings. Thus, where it leads to no confusion, we will omit the term ``straight-line''. Let $\Gamma$ be a drawing of a graph $G$ and let $H$ be a subgraph of $G$. We denote by $\Gamma[H]$ the restriction of $\Gamma$ to the vertices and edges of $H$.

\paragraph{Planar drawings, graphs, and embeddings.} A drawing of a graph is \emph{planar} if no two edges intersect. A graph is \emph{planar} if it admits a planar drawing. A planar drawing partitions the plane into topologically connected regions, called \emph{faces}. The unique unbounded face is the \emph{outer face}, whereas the remaining faces are the \emph{inner faces}. Two planar drawings of a connected graph are \emph{topologically equivalent} if they have the same circular order of the edges around each vertex and the same cycle bounding the outer face. A \emph{planar embedding} is an equivalence class of planar drawings. A \emph{plane graph} is a planar graph equipped with a planar embedding. In a planar straight-line drawing an internal face (the outer face) is \emph{strictly convex} if its angles are all smaller (greater) than $\pi$. A planar straight-line drawing is \emph{strictly convex} if each face is strictly convex. 

%Whenever we talk about a planar drawing of a plane graph $G$, we always assume, even when not explicitly stated, that the drawing respects the planar embedding associated to $G$. Further, whenever we talk about a subgraph $H$ of a plane graph $G$, we always assume, even when not explicitly stated, that $H$ is associated with the planar embedding obtained from the one associated to $G$ by removing vertices and edges not in $H$.

A \emph{$y$-assignment} $y_G: V(G) \rightarrow \mathbb{R}$ is an assignment of 
reals to the vertices of a graph $G$. 
A drawing $\Gamma$ of $G$ \emph{satisfies} $y_G$ if the $y$-coordinate in $\Gamma$ of each vertex $v \in V(G)$ is $y_G(v)$. 
An \emph{$x$-assignment} $x_G$ for the vertices of $G$ is defined analogously.

\paragraph{Connectivity.} A {\em $k$-cut} in a connected graph $G$ is a set of $k$ vertices whose removal disconnects $G$. A graph is {\em $k$-connected} if it does not contain any $(k-1)$-cut; $2$-connected and $3$-connected graphs are also called {\em biconnected} and {\em triconnected} graphs, respectively. The maximal biconnected subgraphs of a graph are called {\em blocks}. A biconnected plane graph $G$ is \emph{internally $3$-connected} if, for every $2$-cut $\{u,v\}$, $u$ and $v$ are incident to the outer face of $G$ and each connected component of the graph $G - \{u,v\}$ contains a vertex incident to the outer face of $G$. 

\paragraph{Directed graphs.}
In a directed graph $G$ we denote by $uv$ an edge directed from a vertex $u$ to a vertex $v$; then $v$ is a \emph{successor} of $u$, and $u$ is a \emph{predecessor} of $v$. A {\em source} is a vertex with no incoming edge; a {\em sink} is a vertex with no outgoing edge. A {\em directed path} consists of the edges $u_iu_{i+1}$, for $i=1,\dots,n-1$.  A {\em directed cycle} consists of the edges $u_iu_{i+1}$, for $i=1,\dots,n$, where $u_{n+1}=u_1$. A graph without directed cycles is {\em acyclic}. A {\em transitive edge} in a directed graph $G$ is an edge $uv$ such that $G$ contains a directed path from $u$ to $v$ different from the edge $uv$.  A \emph{reduced} graph is a directed graph that does not contain any transitive edges. Whenever we do not know or are not interested in the orientation of an edge connecting two vertices $u$ and $v$, we denote it by $(u,v)$. The \emph{underlying graph} of a directed graph $G$ is the undirected graph obtained from $G$ by omitting the directions from its edges. When talking about the connectivity of a directed graph we always refer to the connectivity of its underlying graph. A \emph{topological ordering} of an $n$-vertex acyclic graph $G=(V,E)$ is a numbering $\pi: V \rightarrow \{1,2,\dots,n\}$ of the vertices of $G$ such that $\pi(u) < \pi(v)$, for each edge $uv \in E$.

\paragraph{Upward planar drawings, embeddings, and morphs.}
A drawing of a directed graph is \emph{upward planar} if it is planar and each edge $uv$ is drawn as a curve monotonically increasing in the $y$-direction from $u$ to $v$.  A directed graph is \emph{upward planar} if it admits an upward planar drawing.

Consider an upward planar drawing $\Gamma$ of a directed graph $G$ and consider a vertex $v$. The list $\mathcal S(v)=[w_1,\dots,w_k]$ contains the successors of $v$ in ``left-to-right order''. That is, consider a half-line $\ell$ starting at $v$ and directed leftwards; rotate $\ell$ around $v$ in clockwise direction and append a vertex $w_i$ to $\mathcal S(v)$ when $\ell$ overlaps with the tangent to the edge $(v,w_i)$. The list $\mathcal P(v)=[z_1,\dots,z_l]$ of the predecessors of $v$ is defined similarly. Then two upward planar drawings of a connected directed graph are \emph{topologically equivalent} if they have the same lists $\mathcal S(v)$ and $\mathcal P(v)$ for each vertex $v$. An \emph{upward planar embedding} is an equivalence class of upward planar drawings. An {\em upward plane graph} is an upward planar graph equipped with an upward planar embedding.  If a vertex $v$ in an upward planar graph $G$ is not a source or a sink, then a planar embedding of $G$ determines $\mathcal S(v)$ and $\mathcal P(v)$. However, if $v$ is a source or a sink, then different upward planar drawings might have different lists $\mathcal S(v)$ or $\mathcal P(v)$, respectively. In fact, two upward planar drawings of an upward planar graph $G$ might not have the same upward planar embedding although the underlying graph of $G$ has the same planar embedding in the two drawings; see, for example, \blue{Fig.~\ref{fig:planar-upward-embedding}}.

%	\begin{figure}[tb]
%		\centering
%		\subfloat[]{\includegraphics[scale=1]{Embedding1.pdf}\label{fig:planar-upward-embedding1}}
%		\subfloat[]{\includegraphics[scale=1]{Embedding2.pdf}\label{fig:planar-upward-embedding2}}
%		\caption{Two upward planar drawings of the same directed graph with the same planar embedding but with different upward embeddings. }\label{{fig:planar-upward-embedding}
%	\end{figure}

\begin{figure}[tb]
	\centering
	\subfloat[]
	{\includegraphics[page=1,width=.28\columnwidth]{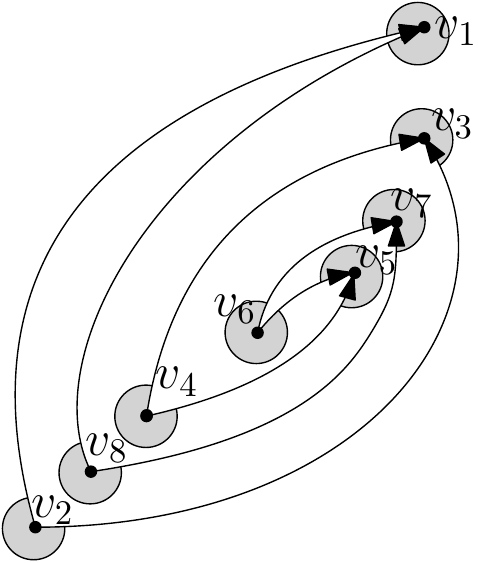}\label{fig:planar-upward-embedding-0}}
	\hfil
	\subfloat[]
	{\includegraphics[page=1,width=.28\columnwidth]{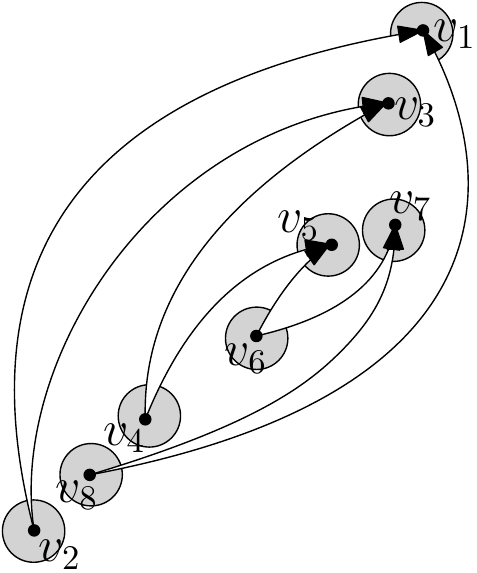}\label{fig:planar-upward-embedding-a}}
		\caption{Two upward planar drawings of the same directed graph $G$ (whose underlying graph is a simple cycle) with the same planar embedding but with different upward planar embeddings. The angles labeled \texttt{large} are gray. Observe that $\mathcal S(v_8)=[v_1,v_7]$ in (a), while $\mathcal S(v_8)=[v_7,v_1]$ in (b).}\label{fig:planar-upward-embedding}
\end{figure}

For biconnected upward planar graphs a different, and yet equivalent, notion of upward planar embedding exists; this is described in the following. Consider an upward planar drawing $\Gamma$ of a biconnected upward planar graph $G$.
Let $u, v$, and $w$ be three distinct vertices that appear consecutively and in this clockwise order along the boundary of a face $f$ of $G$; note that, since $G$ is biconnected $f$ is delimited by a simple cycle. We denote by $\angle (u,v,w)$ the angle formed by the tangents to the edges $(u,v)$ and $(v,w)$ at $v$ in the interior of $f$. We say that $v$ is a \emph{sink-switch} (a \emph{source-switch}) of $f$ if the orientations of the edges $(u,v)$ and $(v,w)$ in $G$ are $uv$ and $wv$ ($vu$ and $vw$, respectively). 
Further, we say that $v$ is a \emph{switch} of $f$ if it is either a sink-switch or a source-switch of $f$, and $v$ is a \emph{switch} of $G$ if it is a switch of some face of $G$. Two switches $u$ and $v$ of a face $f$ are \emph{clockwise} (\emph{counter-clockwise}) \emph{consecutive} if traversing $f$ clockwise (counter-clockwise) no switch is encountered in between $u$ and $v$. 
The drawing $\Gamma$ determines a \emph{large-angle assignment}, that is, a labeling, for each face $f$ and each three clockwise consecutive switches $u$, $v$, and $w$ for $f$ of the corresponding angle $\angle (u,v,w)$ as \texttt{large}, if it is larger than $\pi$ in $\Gamma$, or \texttt{small}, it is smaller than $\pi$  in~$\Gamma$~\cite{DBLP:journals/algorithmica/BertolazziBLM94}. 
Two upward planar drawings of an upward planar graph $G$ are then say to be {\em topologically equivalent} if they have the same planar embedding and the same large-angle assignment. From this notion of topological equivalence, the ones of upward planar embedding and upward plane graph can be introduced as before; again, the formerly introduced notions coincide with the just introduced ones for upward planar graphs with biconnected underlying graphs (in fact, this correspondence between the two notions could be stated for all upward planar graphs, however the definition of clockwise consecutive switches we introduced is ambiguous for upward planar graphs whose underlying graph is not biconnected). A combinatorial characterization of the \texttt{large}-\texttt{small} assignments that correspond to upward planar embeddings is given in~\cite{DBLP:journals/algorithmica/BertolazziBLM94}. 

Whenever we talk about an upward planar drawing of an upward plane graph $G$, we always assume, even when not explicitly stated, that the drawing respects the upward planar embedding associated to $G$. Further, whenever we talk about a subgraph $H$ of an upward plane graph $G$, we always assume, even when not explicitly stated, that $H$ is associated with the upward planar embedding obtained from the one associated to $G$ by removing vertices and edges not in $H$.

Let $\Gamma_0$ and $\Gamma_1$ be two topologically-equivalent upward planar drawings of an upward plane graph $G$. An \emph{upward planar morph} is a continuous transformation from $\Gamma_0$ to $\Gamma_1$ indexed by time $t \in [0,1]$ in which the drawing $\Gamma_t$ at each time $t \in [0,1]$ is upward planar. We remark that each drawing $\Gamma_t$ has to respect the upward planar embedding associated to $G$; in particular, the drawings $\Gamma_0$ and $\Gamma_1$ determine the same upward planar embedding (if this were not the case, then a morph that preserves upward planarity at all times would not exist). 

%Observe that two drawings $\Gamma_0$ and $\Gamma_1$ of the same upward planar graph with different upward planar embeddings do not admit an upward planar morph.

%A morph is composed of one or more \emph{morphing steps}. 
%In each step, a subset of the vertices and their incident edges move,
%along \emph{trajectories} of different types, to a new location. 
%Observe that a morphing step is also a morph.
%A \emph{linear} morphing step is such that vertices move along straight-line trajectories at
%uniform speed. A \emph{unidirectional} morphing step is a linear morphing step
%in which vertex trajectories are all parallel. A \emph{unidirectional} morph is composed 
%by unidirectional morphing steps.
%%

\paragraph{Plane $st$-graphs.} 
A \emph{plane $st$-graph} is an upward plane graph with a single source $s$ and a single sink $t$, and with an upward planar embedding in which $s$ and $t$ are incident to the outer face. A plane $st$-graph always admits an upward planar straight-line drawing~\cite{DBLP:journals/tcs/BattistaT88}. A cycle in an upward plane graph is an \emph{$st$-cycle} if it consists of two directed paths. A face $f$ of an upward plane graph is an \emph{$st$-face} if it is delimited by an $st$-cycle; the directed paths delimiting an $st$-face $f$ are called \emph{left} and \emph{right boundary}, where the edge of the left boundary incident to the source-switch $s_f$ of $f$ immediately precedes the edge of the right boundary incident to $s_f$ in the clockwise order of the edges incident to $s_f$. The following is well-known.

\begin{lemma}\label{le:st-faces-iff-st-graph}
An upward plane graph is a plane $st$-graph iff all its faces are $st$-faces.
\end{lemma}

An internal vertex $v$ of a maximal plane $st$-graph $G$ is \emph{simple} if the neighbors of $v$ induce a cycle in the underlying graph of $G$.

\begin{lemma}[Alamdari et al.~\cite{DBLP:journals/siamcomp/AlamdariABCLBFH17}]\label{le:internal-vertex}
Any maximal plane $st$-graph contains a simple vertex of degree at most $5$.
\end{lemma}

\section{Slow Morphs and Fast Morphs}\label{se:slow-morphs-fast-morphs}

We start this section by proving the following lower bound.

%%%%%%%
%%%%%%%
%%%%%%%
%%%%%%%
%%%%%%%
%\subsection{A Linear Lower Bound}

\begin{figure}[tb]
\centering
\subfloat[$P$]
{\includegraphics[page=4,width=.28\columnwidth]{lower_bound.pdf}\label{fig:lower-bound-0}}
\hfil
\subfloat[$\Gamma_0$]
{\includegraphics[page=2,height=.28\columnwidth]{lower_bound.pdf}\label{fig:lower-bound-a}}
\hfil
\subfloat[$\Gamma_1$]
{\includegraphics[page=1,height=.28\columnwidth]{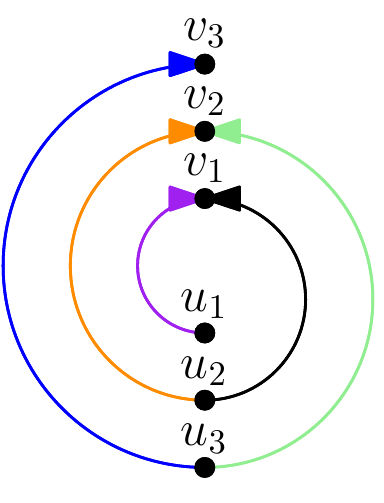}\label{fig:lower-bound-b}}
\caption{Illustration for \autoref{th:lower-bound}. (a) $P$; (b) $\Gamma_0$; and (c) $\Gamma_1$. For the sake of readability $\Gamma_0$ and $\Gamma_1$ have curved edges. However, the $x$-coordinates of the vertices can be slightly perturbed in order to make $\Gamma_0$ and $\Gamma_1$ straight-line.}\label{fig:lower-bound}
\end{figure}

\newcommand{\thlowerbound}{There are two upward planar drawings of an $n$-vertex upward plane path such that any upward planar morph between them consists of $\Omega(n)$ steps.}

\begin{theorem}\label{th:lower-bound}
\thlowerbound
\end{theorem}

\begin{proof}
Assume, for the sake of simplicity, that $n$ is even, and let $n=2k$. Consider the $n$-vertex upward plane path $P$ defined as follows (refer to \blue{Fig.~\ref{fig:lower-bound-0}}).
The path $P$ contains vertices $u_i$ and $v_i$, for $i=1,\dots,k$, and directed edges $u_iv_i$, for $i=1,\dots,k$, and $u_{i+1}v_i$, for $i=1,\dots,k-1$. Clearly, $P$ has a unique planar embedding $\mathcal E$;  we fix the upward planar embedding of $P$ so that $\mathcal S(u_i)=[v_i,v_{i-1}]$, for $i=2,\dots,k$, and so that $\mathcal P(v_i)=[u_i,u_{i+1}]$, for $i=1,\dots,k-1$. 

Let $\Gamma_0$ and $\Gamma_1$ be two upward planar straight-line drawings of $P$ in which the bottom-to-top order of the vertices is $u_1, \dots,u_k,v_k,\dots,v_1$ (see \blue{Fig.~\ref{fig:lower-bound-a}}) and $u_k, \dots,u_1,v_1,\dots,v_k$ (see \blue{Fig.~\ref{fig:lower-bound-b}}), respectively. Note that, by the upward planarity of $\Gamma_0$, the edge $u_iv_i$ has the edge $u_{i+1}v_{i+1}$ to its right in $\Gamma_0$, for $i=1,\dots,k-1$, and the edge $u_{i+1}v_i$ has the edge $u_{i+2}v_{i+1}$ to its left in $\Gamma_0$, for $i=1,\dots,k-2$. Let $\langle \Gamma_0 = \Lambda_1,\Lambda_2,\dots,\Lambda_{h+1}=\Gamma_1 \rangle$ be any upward planar morph from $\Gamma_0$ to $\Gamma_1$ that consists of $h$ morphing steps. We have the following.

\begin{claim}\label{claim:lower}
For each $j=1,2,\dots,\min\{h+1,k-1\}$, we have that:
\begin{enumerate} [(a)]
	\item the vertices $u_{j},u_{j+1},\dots,u_{k-1},u_k,v_k,v_{k-1},\dots,v_{j+1},v_j$ appear in this bottom-to-top order in $\Lambda_{j}$; 
	\item for $i=j,\dots,k-1$, the edge $u_iv_i$ has the edge $u_{i+1}v_{i+1}$ to its right; and 
	\item for $i=j,\dots,k-2$, the edge $u_{i+1}v_i$ has the edge $u_{i+2}v_{i+1}$ to its left. 
\end{enumerate}
\end{claim}

\begin{claimproof}
We prove the statement by induction on $j$. The statement is trivial for $j=1$, by the definition of $\Gamma_0 = \Lambda_1$. 

Consider now any $j>1$. By induction, $\Lambda_{j-1}$ satisfies Properties~(a)--(c). 

Suppose, for a contradiction, that there exists an index $i\in\{j,j+1,\dots,k-1\}$ such that $u_{i+1}$ lies below $u_i$ in $\Lambda_{j}$. The upward planarity of $\Lambda_{j-1}$ and $\Lambda_{j}$ implies that $v_i$ and $v_{i+1}$ both lie above $u_i$, both in $\Lambda_{j-1}$ and $\Lambda_{j}$. Further, $u_{i+1}$ lies below $v_i$ and $v_{i+1}$, both in $\Lambda_{j-1}$ and $\Lambda_{j}$; this comes from Property~(a) of $\Lambda_{j-1}$ and from the assumption that $u_{i+1}$ lies below $u_i$ in $\Lambda_{j}$. Then $u_{i+1}$ lies below the horizontal line through the lowest of $v_i$ and $v_{i+1}$ throughout the linear morph $\langle \Lambda_{j-1},\Lambda_j\rangle$. By Properties~(b) and~(c) of $\Lambda_{j-1}$, the vertex $u_{i+1}$ lies in $\Lambda_{j-1}$ inside the bounded region of the plane delimited by the edge $u_iv_i$, by the edge $u_iv_{i+1}$, and by the horizontal line through the lowest of $v_i$ and $v_{i+1}$. However, by the assumption that $u_{i+1}$ lies below $u_i$ in $\Lambda_{j}$, we have that $u_{i+1}$ lies outside the same region in $\Lambda_{j}$. Since $u_{i+1}$ does not cross the horizontal line through the lowest of $v_i$ and $v_{i+1}$ throughout the linear morph $\langle \Lambda_{j-1},\Lambda_j\rangle$, it follows that $u_{i+1}$ crosses $u_iv_i$ or $u_iv_{i+1}$ during $\langle \Lambda_{j-1},\Lambda_j\rangle$, a contradiction. 

An analogous proof shows that $v_{i+1}$ lies below $v_i$ in $\Lambda_{j}$, for $i=j,j+1,\dots,k-1$. Property~(a) for $\Lambda_{j}$ follows. Properties~(b) and~(c) follow by Property~(a) and by the upward planarity of $\Lambda_{j}$. This concludes the proof of the claim.
\end{claimproof}

By \autoref{claim:lower} and since $u_k,u_{k-1}$ appear in this bottom-to-top order in $\Gamma_1=\Lambda_{h+1}$, we have that $h+1>k-1$, hence $h \in \Omega(n)$.
\end{proof}

\remove{We set as \texttt{large} the angles $\angle (u_{i},v_i,u_{i+1})$ and $\angle(v_i,u_{i+1},v_{i+1})$, for $i=1,\dots,k-1$.}

%\subsection{Geometric Tools}

% \begin{lemma}[\cite{DBLP:journals/siamcomp/AlamdariABCLBFH17},Lemma~2.1]\label{le:unidirectional-combination}
% Let $x,y,z$ be the clockwise-ordered vertices of the triangular outer face of a straight-line planar drawing. Suppose that vertices $x$, $y$, and $z$ move linearly in the direction of a vector $\vec{\ell}$ in such a way that their clockwise order is preserved. Any point $p$ inside the triangle can be defined as a convex combination of $x$, $y$, and $z$, and in this way the motion of $x$, $y$, and $z$ determines the motion of $p$. The result is a unidirectional morph of the straight-line planar drawing (in particular, planarity is preserved).
% \end{lemma}

% Let $\ell$ be a line and let $\Gamma'$ and $\Gamma''$ be two planar straight-line drawings of the same plane graph. Then, the pair $(\Gamma',\Gamma'')$ have the \emph{sidedness property with respect to $\ell$}~\cite{Barrera-unidirectional} if every oriented line $\ell'$ parallel to $\ell$ either does not intersect any edge in both $\Gamma'$ and in $\Gamma''$ or intersects the same edges in both $\Gamma'$ and $\Gamma''$ and does so in the same order in $\Gamma'$ as in $\Gamma''$.

We now establish a tool that will allow us to design efficient algorithms for morphing upward planar drawings. Consider two planar straight-line drawings $\Gamma'$ and $\Gamma''$ of a plane graph $G$ with the same $y$-assignment. Since the drawings are straight-line and have the same $y$-assignment, a horizontal line $\ell$ intersects a vertex or an edge of $G$ in $\Gamma'$ if and only if it intersects the same vertex or edge in $\Gamma''$. We say that $\Gamma'$ and $\Gamma''$ are {\em left-to-right equivalent} if, for any horizontal line $\ell$, for any vertex or edge $\alpha$ of $G$, and for any vertex or edge $\beta$ of $G$ such that $\ell$ intersects both $\alpha$ and $\beta$ (in $\Gamma'$ and in $\Gamma''$), we have that the intersection of $\alpha$ with $\ell$ is to the left of the intersection of $\beta$ with $\ell$ in $\Gamma'$ if and only if the intersection of $\alpha$ with $\ell$ is to the left of the intersection of $\beta$ with $\ell$ in $\Gamma''$. The definition of {\em bottom-to-top equivalent} drawings is analogous. We have the following.

\begin{lemma}\label{le:st-y-assignment-equivalent}
	Any two upward planar drawings $\Gamma'$ and $\Gamma''$ of a plane $st$-graph $G$ with the same $y$-assignment are left-to-right equivalent.
\end{lemma}
\begin{proof}
	Since $G$ is a plane $st$-graph, the drawings $\Gamma'$ and $\Gamma''$ have the same faces. By \autoref{le:st-faces-iff-st-graph} such faces are $st$-faces. Also, every horizontal line $\ell$ crosses an $st$-face $f$ at most twice, and the left-to-right order of these crossings along $\ell$ is the same in $\Gamma'$ and $\Gamma''$ because the left and right boundaries of $f$ are the same in $\Gamma'$ and $\Gamma''$, given that $\Gamma'$ and $\Gamma''$ are topologically equivalent.
\end{proof}

\autoref{le:sidedness} is due to \cite{DBLP:journals/siamcomp/AlamdariABCLBFH17}. We extend it in \autoref{le:shift}.

\begin{lemma}[Alamdari et al.~\cite{DBLP:journals/siamcomp/AlamdariABCLBFH17}, Corollary 7.2]\label{le:sidedness}
Consider a unidirectional morph acting on points $p$, $q$, and $r$. If
$p$ is on one side of the oriented line through $\overline{qr}$ at the beginning and at the end of the morph, then $p$ is on the same side of the oriented line through $\overline{qr}$ throughout the morph.
\end{lemma}

% \begin{lemma}[\cite{DBLP:journals/siamcomp/AlamdariABCLBFH17}, Corollary 7.2]\label{le:sidedness}
% Consider an unidirectional linear morphing step acting on points $p$, $q$, and $r$. If
% $p$ is to the right of the line through $\overline{qr}$ at the beginning and the end of the morph, then $p$ is to the right of the line through $\overline{qr}$ throughout the morph.
% \end{lemma}

\begin{lemma}\label{le:shift}
	Let $\Gamma'$ and $\Gamma''$ be two left-to-right or bottom-to-top equivalent planar drawings of a plane graph. Then the linear morph $\mathcal M$ from $\Gamma'$ to~$\Gamma''$ is unidirectional and planar.
\end{lemma} 
\begin{proof}
%Consider two planar drawings $\Gamma'$ and $\Gamma''$ of the same plane graph that are left-to-right (bottom-to-top) equivalent and that have the same $y$-assignment ($x$-assignment). Consider the morph ${\mathcal M}$ from $\Gamma'$ and $\Gamma''$. Morph ${\mathcal M}$ is unidirectional since all the vertices move along horizontal (vertical) trajectories. 
%Suppose for a contradiction that ${\mathcal M}$ is not planar. Then there exists at least a vertex $u$ that intersects an edge $(v,w)$ during the morph. However, since $\Gamma'$ and $\Gamma''$ are left-to-right equivalent, we have that $u$ is on the same side of $(v,w)$ both in $\Gamma'$ and in $\Gamma''$, contradicting \cite[Lemma~5]{Barrera-unidirectional,DBLP:journals/corr/Barrera-CruzHL14}.
%
Since $\Gamma'$ and $\Gamma''$ have the same $y$-assignment ($x$-assignment), given that they are left-to-right (bottom-to-top) equivalent, it follows that all the vertices move along horizontal (vertical) trajectories. Thus, $\mathcal M$ is unidirectional. 
Also, since $\Gamma'$ and $\Gamma''$ are left-to-right (bottom-to-top) equivalent, each horizontal (vertical) line crosses the \mbox{same sequence} of vertices and edges in both $\Gamma'$ and $\Gamma''$. Thus, by \autoref{le:sidedness}, $\mathcal M$ is planar.
%
%\red{PROVARE o IMPORTARE?
% It has been proved in~\cite{addfpr-mpgdo-14} that \morph{\Gamma_1,\Gamma_2} is planar and unidirectional. To prove that it is strictly-convex, note that an angle delimited by two edges $(u,v)$ and $(u,z)$ that is convex in $\Gamma_1$ and in $\Gamma_2$ stays convex during the entire morph; this descends from the planarity of \morph{\Gamma_1,\Gamma_2} if $\gamma(u)<\gamma(v),\gamma(z)$ or $\gamma(v),\gamma(z)<\gamma(u)$ and from the fact that $u$, $v$, and $z$ are never aligned during \morph{\Gamma_1,\Gamma_2} if $\gamma(z)<\gamma(u)<\gamma(v)$.
%}
\end{proof}

\autoref{le:shift} allows us to devise a simple morphing technique between any two upward planar drawings $\Gamma_0$ and $\Gamma_1$ of the same upward plane graph $G$, when a pair of upward planar drawings of $G$ with special properties can be computed. We say that the pair $(\Gamma_0,\Gamma_1)$ is an \emph{hvh-pair} if there exist upward planar drawings
$\Gamma'_0$ and $\Gamma'_1$ of $G$ such that:
\begin{inparaenum}[(i)]
\item \label{pr:one} $\Gamma_0$ and $\Gamma'_0$ are left-to-right equivalent,
\item \label{pr:two} $\Gamma'_0$ and $\Gamma'_1$ are bottom-to-top equivalent, and
\item \label{pr:three} $\Gamma'_1$ and $\Gamma_1$ are left-to-right equivalent.
\end{inparaenum}
Our morphing tool is expressed by the following lemma.

\newcommand{\fastmorphstatement}{Let $(\Gamma_0,\Gamma_1)$ be an hvh-pair of upward planar drawings of an upward plane graph $G$.
	There is a $3$-step \mbox{upward planar morph from $\Gamma_0$ to $\Gamma_1$}.}
\begin{lemma}[Fast morph]\label{le:fast-morph}
\fastmorphstatement
\end{lemma}

\begin{proof}
	By hypothesis there exist drawings $\Gamma'_0$ and $\Gamma'_1$ of $G$ satisfying Conditions~(\ref{pr:one}),~(\ref{pr:two}), and~(\ref{pr:three}) of the definition of an hvh-pair.
	By \autoref{le:shift}, $\mathcal M_1 = \langle \Gamma_0, \Gamma'_0 \rangle$, 
	$\mathcal M_2 = \langle \Gamma'_0, \Gamma'_1 \rangle$, and $\mathcal M_3 = \langle \Gamma'_1, \Gamma_1 \rangle$ are planar linear morphs. 
	Therefore, $\mathcal M = \langle  \Gamma_0, \Gamma'_0, \Gamma'_1, \Gamma_1 \rangle$ is a $3$-step planar morph from $\Gamma_0$ to $\Gamma_1$. In order to prove that $\mathcal M$ is an upward planar morph, we need to show that each linear morph $\mathcal M_i$ is an upward planar morph. To this aim, we only need to prove that no edge changes its orientation during $\mathcal M_i$, for $i=1,2,3$. This is trivially true for
	$\mathcal M_1$ (for $\mathcal M_3$) since $\Gamma_0$ and $\Gamma'_0$ ($\Gamma_1$ and $\Gamma'_1$) induce the same $y$-assignment.
	
	We now prove that no directed edge $uv$ changes its orientation during $\mathcal M_2$. 
	%Let $uv$ be a directed edge of $G$ and let $\ell'_0$ and $\ell'_1$ be the vertical lines passing through $u$ in $\Gamma'_0$ and $\Gamma'_1$, respectively. Consider the projections $v_0$ and $v_1$ of vertex $v$ onto $\ell'_0$ and $\ell'_1$, respectively. Since $\Gamma'_0$ and $\Gamma'_1$ are upward planar drawings of $G$, $v_i$ lies above $u$ along $\ell'_i$, $i=0, 1$. Thus, by Lemma~\ref{le:sidedness}, the projection of $v$ remains above $u$ at each time during $\mathcal M_2$, which implies that $uv$ maintains its orientation during $\mathcal M_2$ and concludes the proof. 
	By Condition~(\ref{pr:two}) of the definition of an hvh-pair, the $x$-coordinate of $u$ is the same in $\Gamma'_0$ and in $\Gamma'_1$, hence it is the same throughout $\mathcal M_2$. Denote by $x'$ such $x$-coordinate. The $y$-coordinate of $u$ might be different in $\Gamma'_0$ and in $\Gamma'_1$; denote by $y'_0$ and $y'_1$ such coordinates, respectively. Consider a point $r$ that moves (at uniform speed along a straight-line trajectory) during $\mathcal M_2$ from $(x'+1,y'_0)$ in $\Gamma_0$ to $(x'+1,y'_1)$ in $\Gamma_1$. Note that $r$ moves along a vertical trajectory, hence the movement of $r$ and $\mathcal M_2$ define a unidirectional morph. Also observe that the straight-line segment $\overline{ur}$ is horizontal throughout $\mathcal M_2$; further, $v$ is above the horizontal line through $\overline{ur}$ both in $\Gamma'_0$ and in $\Gamma'_1$, by the upward planarity of $\Gamma_0$ and $\Gamma_1$ and by Conditions~(\ref{pr:one}) and~(\ref{pr:three}) of the definition of an hvh-pair. By Lemma~\ref{le:sidedness} with $p=v$, $q=u$, and $r=r$ we have that the $y$-coordinate of $v$ is greater than the $y$-coordinate of $u$ throughout $\mathcal M_2$. Hence, $\mathcal M_2$ is an upward planar morph.
\end{proof}

%\begin{proofsketch}
%	We define the morph $\mathcal M$ as $\langle  \Gamma_0, \Gamma'_0, \Gamma'_1, \Gamma_1 \rangle$. The drawings $\Gamma'_0$ and $\Gamma'_1$ exist by hypothesis. \autoref{le:shift} guarantees that $\mathcal M$ is unidirectional and planar. We use \autoref{le:sidedness} to prove that $\mathcal M$ is upward.
%\end{proofsketch}

\newcommand{\biconnectedreducedstatement}{	Let $\Gamma_0$ and $\Gamma_1$ be two upward planar drawings of an $n$-vertex upward plane graph $G$ whose underlying graph is connected. There exist upward planar drawings $\Gamma'_0$ and $\Gamma'_1$ of an $O(n)$-vertex upward plane graph $G'$ that is a supergraph of $G$, whose underlying graph is biconnected, and such that $\Gamma'_0[G]=\Gamma_0$ and $\Gamma'_1[G]=\Gamma_1$. Further, if $G$ is reduced or an $st$-graph, then so is $ G'$.}

The next lemma will allow us to restrict our attention to biconnected graphs.

\begin{lemma}\label{le:biconnected-reduced-augmentation}
\biconnectedreducedstatement
\end{lemma}
%\begin{proofsketch}	
%We iteratively apply the following procedure. Consider a cutvertex $v$ of $G$ and two edges that belong to distinct blocks of $G$ and that are consecutive in the circular order of the edges incident to $v$. Let $u$ and $w$ be the end-vertices of such edges different from $v$. We add to $G$ a vertex $v'$ and two edges connecting $v'$ with $u$ and $w$; these edges are oriented as the ones connecting $v$ with $u$ and $w$, respectively. By placing $v'$ and its incident edges inside the face of $G$ incident to $v$, $u$, and $w$, we obtain an upward plane supergraph of $G$ with one block less than~$G$. Upward planar drawings of this graph extending $\Gamma_0$ and $\Gamma_1$ can be easily obtained. The repetition of this procedure proves the lemma.
%%	For each cutvertex $v$ of $G$, we add to $G$ a new vertex $v'$ and two edges connecting $v'$ to two neighbors $u$ and $w$ of $v$ belonging to different blocks and consecutive in the circular order around $v$. By placing $v'$ inside the face of $G$ containing $v$, $u$, and $w$, and by appropriately orienting the newly introduced edges, we obtain an upward plane supergraph of $G$ containing one block less than~$G$. Also, upward planar drawings of this graph that extend $\Gamma_0$ and $\Gamma_1$ can be easily obtained. Thus, by applying this procedure to each cutvertex, the statement follows.
%\end{proofsketch}

\begin{proof}
	Initialize $G'=G$, $\Gamma'_0=\Gamma_0$, and $\Gamma'_1=\Gamma_1$. Consider a cutvertex $v$ of~$G'$.
	Let $u$ and $w$ be two neighbors of $v$ belonging to different blocks of $G'$ that are consecutive in the circular order of the neighbors of $v$. By relabeling $u$ and $w$, we may assume that one of the following holds true:
	\begin{itemize}
		\item if $u$ and $w$ are both successors of $v$, then the edge $vw$ immediately follows the edge $vu$ in the clockwise order of the edges incident to $v$ (see \blue{Fig.}\autoref{fig:bico1}); %, and $vu$ is to the left of $vw$ %and $uv$ is to the left of $wv$
		\item if $u$ and $w$ are both predecessors of $v$, then the edge $wv$ immediately precedes the edge $uv$ in the clockwise order of the edges incident to $v$ (see \blue{Fig.}\autoref{fig:bico2}); or
		\item $u$ is a successor of $v$, $w$ is a predecessor of $v$, and the edge $wv$ immediately follows the edge $vu$ in the clockwise order of the edges incident to $v$ (see \blue{Fig.}\autoref{fig:bico3}).
	\end{itemize}

	\begin{figure}[tb]
		\centering
		\subfloat[]{\includegraphics[scale=1, page=1]{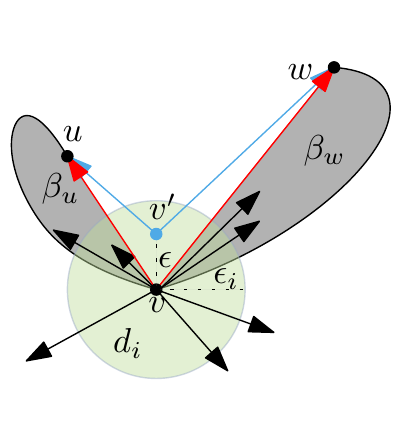}\label{fig:bico1}}
		\subfloat[]{\includegraphics[scale=1, page=2]{biconnection}\label{fig:bico2}}
		\subfloat[]{\includegraphics[scale=1, page=3]{biconnection}\label{fig:bico3}}
		\caption{Illustration for \autoref{le:biconnected-reduced-augmentation}. The vertices $u$ and $w$ belong to two blocks $\beta_u$ and $\beta_w$, respectively, both containing the cut-vertex $v$. \label{fig:bico}}
	\end{figure}
	
	Denote by $f$ the face that is to the right of the edge $(u,v)$ when traversing such an edge according to its orientation. Note that the edge $(v,w)$ is also incident to $f$. We add to $G'$ a new vertex $v'$ inside $f$; further, we add to $G'$ two directed edges connecting $v'$ with $u$ and $w$ inside $f$. These edges are directed as the edges connecting $v$ with $u$ and $w$, respectively; that is, we add to $G'$ either the directed edge $uv'$, if $uv \in E(G')$, or the directed edge $v'u$, if $vu \in E(G')$, and either the directed edge $wv'$, if $wv \in E(G')$, or the directed edge $v'w$, if $vw \in E(G')$. 
	
	The described augmentation does not introduce any transitive edges. Further, no edge that was already in $G'$ before the augmentation becomes transitive after the augmentation; this is because the edge $(u,w)$ does not belong to $G'$, as $u$ and $w$ belong to distinct blocks of $G'$; hence $G'$ remains reduced if it was so. 
	
	In the case illustrated in \blue{Fig.}\autoref{fig:bico1} (in \blue{Fig.}\autoref{fig:bico2}), each of the blocks $\beta_u$ and $\beta_w$ of $G'$ containing $u$ and $w$ before the augmentation contains a distinct sink of $G'$ (resp.\ a distinct source of $G'$), hence $G'$ is not an $st$-graph before the augmentation. In the case illustrated in \blue{Fig.}\autoref{fig:bico3}, it might be that $G'$ is an $st$-graph before the augmentation. Note that there only two faces of the augmented graph $G'$ that do not belong to $G'$ before the augmentation. One of them is delimited by the directed paths $wvu$ and $wv'u$, hence it is an $st$-face; the other one is obtained from $f$ by replacing the directed path $wvu$ with the directed path $wv'u$, hence it is an $st$-face as long as $f$ is. It follows that, if $G'$ is an $st$-graph before the augmentation, then it remains an $st$-graph after the augmentation. 
	
	We now describe how to insert $v'$ and its incident edges into $\Gamma'_0$ and $\Gamma'_1$. By standard continuity arguments, like the ones used in the proof of F\'ary's theorem~\cite{Fary}, we have that, for $i=0,1$, there exists a sufficiently small value $\epsilon_i > 0$ such that the disk $d_i$ with radius  $\epsilon_i$ centered at $v$ in $\Gamma'_i$ contains no vertex other than $v$ and is not traversed by any edge other than those incident to $v$. We place $v'$ at distance $\epsilon < \epsilon_0,\epsilon_1$ from $v$ inside $f$, as illustrated in \blue{Figs.}\autoref{fig:bico1}--\autoref{fig:bico3}; in particular, $v'$ is placed in the circular sector of $d_i$ delimited by $(u,v)$ and $(w,v)$. By selecting a sufficiently small value for $\epsilon$, the edges $(u,v')$ and $(w,v')$ can be drawn as straight-line segments that do not intersect any edge of $G'$. Further, if $\epsilon$ is sufficiently small, then the $y$-coordinate of $u$ (the $y$-coordinate of $w$) is smaller than the one of $v'$ if and only if it is smaller than the one of $v$, hence the straight-line segments representing the edges $(u,v')$ and $(w,v')$ monotonically increase in the $y$-direction from their sources to their sinks. The upward planarity of the drawings $\Gamma'_0$ and $\Gamma'_1$ of the augmented graph $G'$ follows. Note that after the augmentation we have $\Gamma'_0[G]= \Gamma_0$ and $\Gamma'_1[G]= \Gamma_0$. This is because the same equalities were satisfied before the augmentation and since the drawings of $G'$ before the augmentation were not altered during the augmentation.
	
	Since the graph $G'$ after the augmentation contains one block less than before the augmentation, the repetition of this argument results in a biconnected graph $G'$. This concludes the proof of the lemma.
\end{proof}

%%%%%%%%%%%%%%%%%%%%%%%%%%%%%%%%%%%%%%%%%%%%%%%%%%%%%%%%%%%%%%%%%%%%%%%%%%%%%%%
%%%%%%%%%%%%%%%%%%%%%%%%%%%%%%%%%%%%%%%%%%%%%%%%%%%%%%%%%%%%%%%%%%%%%%%%%%%%%%%
%%%%%%%%%%%%%%%%%%%%%%%%%%%%%%%%%%%%%%%%%%%%%%%%%%%%%%%%%%%%%%%%%%%%%%%%%%%%%%%%     _                              _
% ___| |_       __ _ _ __ __ _ _ __ | |__  ___
%/ __| __|____ / _` | '__/ _` | '_ \| '_ \/ __|
%\__ \ ||_____| (_| | | | (_| | |_) | | | \__ \
%|___/\__|     \__, |_|  \__,_| .__/|_| |_|___/
%              |___/          |_|
%%%%%%%%%%%%%%%%%%%%%%%%%%%%%%%%%%%%%%%%%%%%%%%%%%%%%%%%%%%%%%%%%%%%%%%%%%%%%%%%%%%%%%%%%%%%%%%%%%%%%%%%%%%%%%%%%%%%%%%%%%%%%%%%%%%%%%%%%%%%%%%%%%%%%%%%%%%%%%
%%%%%%%%%%%%%%%%%%%%%%%%%%%%%%%%%%%%%%%%%%%%%%%%%%%%%%%%%%%%%%%%%%%%%%%%%%%%%%%
%%%%%%%%%%%%%%%%%%%%%%%%%%%%%%%%%%%%%%%%%%%%%%%%%%%%%%%%%%%%%%%%%%%%%%%%%%%%%%%

\section{Plane $st$-Graphs}\label{se:st-graphs}

In this section, we show algorithms for constructing upward planar morphs between upward planar drawings of plane $st$-graphs.

%We consider reduced (\autoref{sse:reduced}) 
%%, bitonic triangulations (Section~\ref{sse:bitonic}), 
%and general $st$-graphs (\autoref{sse:general-st-graph}).
%We show that the graphs in the first class allow for upward planar morphs with a constant number of steps (\autoref{th:st-reduced-graphs}), whereas upward planar morphs with a linear number of steps always suffice for graphs in the second class (\autoref{le:st-graphs-maximal}).

\subsection{Reduced Plane $st$-Graphs}\label{sse:reduced}

We first consider plane $st$-graphs without transitive edges. We have the following.

\newcommand{\goodreducedstatement}{Any two upward planar drawings $\Gamma_0$ and $\Gamma_1$ of a reduced plane $st$-graph $G$ form an hvh-pair.}
\begin{lemma}\label{le:good-reduced}
\goodreducedstatement
\end{lemma}

\begin{proof}

By \autoref{le:biconnected-reduced-augmentation} we can assume that the reduced plane $st$-graph $G$ is biconnected. Let $\Gamma_0$ and $\Gamma_1$ be any two upward planar drawings of $G$. We show that $\Gamma_0$ and $\Gamma_1$ form an hvh-pair by exhibiting two upward planar drawings $\Gamma'_0$ and $\Gamma'_1$ of $G$ that satisfy Conditions~(\ref{pr:one}),~(\ref{pr:two}), and~(\ref{pr:three}) of the definition of an hvh-pair.
	
	We construct drawings $\Gamma'_0$ and $\Gamma'_1$ as follows (refer to \blue{Figs.}~\ref{fig:ear-decomposition} and \ref{fig:ear-inv}).
	Consider the {\em weak dual} multi-graph $D$ of $G$, which is defined as follows. The multi-graph $D$ has a vertex $v_f$ for each internal face $f$ of $G$ and a directed edge $v_f v_g$ if the faces $f$ and $g$ of $G$ share an edge $e$ in $G$ and $f$ lies to the left of $g$ when traversing $e$ according to its orientation. The concept of weak dual multi-graph has been used, e.g., in~\cite{RT-rpl-86,TT86-avrpg-86,TT-uavr-86}. Observe that $D$ is acyclic~\cite{TT86-avrpg-86}. We now present a structural decomposition of $G$ guided by $D$ which has been used, e.g., in~\cite{fgw-nupo-12,Mel}.
	Let $\mathcal{T}=\{v_1, \dots, v_k\}$ be a topological ordering of the vertices of $D$ and let $P_0$ be the left boundary of the outer face of $G$. The ordering $\mathcal{T}$ defines a sequence $P_1$, $P_2$, \dots, $P_k$ of directed paths such that, for each $i=1,\dots,k$, the path $P_j$ is the right boundary of the face of $G$ corresponding to the vertex $v_j$ of~$D$. For $j=1,\dots,k$, the graph $G_j=\bigcup_{i=0}^j P_i$ is a plane $st$-graph which is obtained by attaching the directed path $P_{j}$ to two non-adjacent vertices on the right boundary of the outer face of $G_{j-1}$; further, $G_k=G$. Note that, since $G$ is a reduced plane $st$-graph, no path $P_j$ consists of a single edge.
	
	The drawings $\Gamma'_0$ and $\Gamma'_1$ are simultaneously and iteratively constructed by adding, for $j=1,\dots,k$, the path $P_j$ to the already constructed drawings $\Gamma'_0$ and $\Gamma'_1$ of $G_{j-1}$. Note that, after $P_j$ has been drawn in $\Gamma'_0$ and $\Gamma'_1$, the right boundary of the outer face of $G_{j}$ is a directed path, hence it is represented by a $y$-monotone curve in both $\Gamma'_0$ and $\Gamma'_1$.
	
	\begin{figure}[tb]
		\centering
		
		\subfloat[]{\includegraphics[scale=1, page=1]{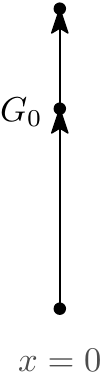}\label{fig:ear-deco-p0}}
		\hfil
		\subfloat[]{\includegraphics[scale=1, page=2]{ear-decomposition}\label{fig:ear-deco-p1}}
		% \hfil
		% \subfloat[]{\includegraphics[scale=1, page=3]{ear-decomposition}\label{fig:ear-deco-lines}}
		% \hfil
		% \subfloat[]{\includegraphics[scale=1, page=4]{ear-decomposition}\label{fig:ear-deco-lines-pj}}
		% \hfil
		% \subfloat[]{\includegraphics[scale=1, page=5]{ear-decomposition}\label{fig:ear-deco-gj}}
		% \hfil
		\caption{Illustration for the proof of \autoref{le:good-reduced}. \protect\subref{fig:ear-deco-p0} The drawing $\Gamma'_i$ of $G_0=P_0$. \protect\subref{fig:ear-deco-p1} The drawing $\Gamma'_i$ of  $G_1$. %\protect\subref{fig:ear-deco-lines} Identification of $x^*$ for $P_j$ in $\Gamma'_{j-1}$ . 
			% 	\protect\subref{fig:ear-deco-lines-pj} Insertion of $P_j$. \protect\subref{fig:ear-deco-gj} Drawing $\Gamma'_j$ of $G_j$.
			\label{fig:ear-decomposition}
		}
	\end{figure}
	
	\begin{figure}[tb]
		\centering
		\subfloat[]{\includegraphics[scale=1, page=1]{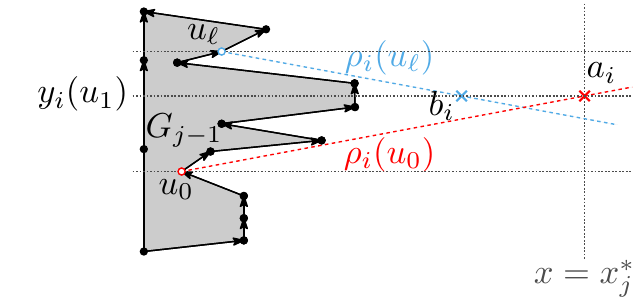}\label{fig:ear-inv-1}}
		\subfloat[]{\includegraphics[scale=1, page=2]{ear-decomposition-invariant}\label{fig:ear-inv-2}}
		\\
		\subfloat[]{\includegraphics[scale=1, page=3]{ear-decomposition-invariant}\label{fig:ear-inv-3}}
		\subfloat[]{\includegraphics[scale=1, page=4]{ear-decomposition-invariant}\label{fig:ear-inv-4}}
		\caption{Illustration for the proof of \autoref{le:good-reduced}. 
			Computation of the points $a_i$ and $b_i$ in $\Gamma'_i$ for the cases $\ell=2$ \protect\subref{fig:ear-inv-1} and $\ell>2$ \protect\subref{fig:ear-inv-3}, respectively. Drawing of the path $P_j$ for the cases $\ell=2$ \protect\subref{fig:ear-inv-2} and $\ell>2$ \protect\subref{fig:ear-inv-4}.
			\label{fig:ear-inv}
		}
	\end{figure}
	
	For $i=0,1$, we denote by $y_i(v)$ the $y$-coordinate of a vertex $v$ in $\Gamma_i$.
	
	We obtain drawings $\Gamma'_0$ and $\Gamma'_1$ of $G_0 = P_0$ by placing its vertices along the line $x=0$ at the $y$-coordinates they have in $\Gamma_0$ and $\Gamma_1$, respectively, and by drawing its edges as straight-line segments (see \blue{Fig.}\autoref{fig:ear-deco-p0}). It is easy to see that $\Gamma_0[G_0]$, $\Gamma'_0$, $\Gamma'_1$, and $\Gamma_1[G_0]$ fulfill Conditions~(\ref{pr:one})-(\ref{pr:three}) of the definition of an hvh-pair.
	
	Suppose now that, for some $j\in \{1,\dots,k\}$, the drawings $\Gamma'_0$ and $\Gamma'_1$ are upward planar straight-line drawings of $G_{j-1}$ such that $\Gamma_0[G_{j-1}]$, $\Gamma'_0$, $\Gamma'_1$, and $\Gamma_1[G_{j-1}]$ fulfill Conditions~(\ref{pr:one})-(\ref{pr:three}) of the definition of an hvh-pair.
	%, for each vertex $v$ of $G_{j-1}$, the $y$-coordinate of $v$ in $\Gamma'_i$ is the same as in $\Gamma_i$, and the $x$-coordinate of $v$ is the same in $\Gamma'_0$ and $\Gamma'_1$. 
	
	We show how to add the path $P_j= u_0 u_1 \dots u_{\ell-1} u_\ell$ to both $\Gamma'_0$ and $\Gamma'_1$ so that the resulting drawings together with $\Gamma_0[G_{j}]$ and $\Gamma_1[G_{j}]$ fulfill Conditions~(\ref{pr:one})-(\ref{pr:three}) of the definition of an hvh-pair. Note that $u_0$ and $u_\ell$ belong to the right boundary of the outer face of $G_{j-1}$, hence they are already present in $\Gamma'_0$ and $\Gamma'_1$. Since all the edges of $P_j$ are going to be drawn as straight-line segments, it suffices to show how to draw the internal vertices of $P_j$ in $\Gamma'_0$ and $\Gamma'_1$.
	For $i=0,1$, we assign to the internal vertices of $P_j$ in $\Gamma'_i$ the same $y$-coordinates they have in $\Gamma_i$. Also, we assign to all such vertices, in both drawings, the same $x$-coordinate $x^*_j$, which has a ``sufficiently large'' value determined as follows. 
	
	If $j=1$, then we set $x^*_j=1$ (see \blue{Fig.}~\ref{fig:ear-deco-p1}). 
	
	If $j>1$, then we proceed as follows. Refer to~\autoref{fig:ear-inv}.
		%Given two lines $L'$ and $L''$ we denote by $\lceil L',L''\rceil$ and by $\lfloor L',L''\rfloor$ the largest and the smallest slope, respectively, among the slopes of  $L'$  and $L''$.
		For $i=0,1$, let $\rho_i(u_0)$ be a ray emanating from $u_0$ with positive slope, directed rightwards, not intersecting $\Gamma'_i$, except at $u_0$, and such that the intersection point $a_i$ of $\rho_i(u_0)$ with the horizontal line $y= y_i(u_1)$ lies to the right of every vertex in $\Gamma'_i$. Analogously, for $i=0,1$, let $\rho_i(u_\ell)$ be a ray emanating from $u_\ell$ with negative slope, directed rightwards, not intersecting $\Gamma'_i$, 
		except at $u_\ell$, and such that the intersection point $b_i$ of $\rho_i(u_\ell)$ with the horizontal line $y= y_i(u_{\ell-1})$ lies to the right of every vertex in $\Gamma'_i$. Refer to \blue{Figs.}~\ref{fig:ear-inv-1} and~\ref{fig:ear-inv-3} for the cases in which $\ell=2$ and $\ell>2$, respectively.
		Observe that $u_1$ and $u_{\ell-1}$ coincide if $P_j$ is a path of length $2$, i.e., if $\ell = 2$.
		We set $x^*_j$ as the maximum of the $x$-coordinates of $a_0$, $b_0$, $a_1$, and $b_1$.

% Otherwise, let $\ell_1$ be the straight-line passing through $u_0$ with the smallest positive slope and intersecting at least a vertex of $G_{j-1}$ in $\Gamma'_i$, for $i\in \{0, 1\}$. Analogously, let $\ell_2$ be the straight-line passing through $u_\ell$ with the largest negative slope and intersecting at least a vertex of $G_{j-1}$ in $\Gamma'_i$, for $i\in \{0, 1\}$. Let $p$ be the point at which lines $\ell_1$ and $\ell_2$ cross (see \blue{Fig.}~\ref{fig:ear-deco-lines}). Denote by $p_1$ the point of  $\ell_1$ whose $y$-coordinate is the maximum between those assigned to vertex $u_1$ in $\Gamma_0$ and $\Gamma_1$. Analogously, denote by $p_2$ the point of  $\ell_2$ whose $y$-coordinate is the minimum between those assigned to vertex $u_{l-1}$ in $\Gamma_0$ and $\Gamma_1$. We set $x^*$ as the minimum integer greater than the maximum $x$-coordinate between those of $p$, $p_1$, and $p_2$ (see \blue{Fig.}~\ref{fig:ear-deco-lines-pj}).
	
	The obtained drawings $\Gamma'_0$ and $\Gamma'_1$ of $G_j$ are planar, as the slopes of the straight-line segments representing the edge $u_0 u_1$ in $\Gamma'_0$ and $\Gamma'_1$ are smaller than or equal to those of $\rho_0(u_0)$ and $\rho_1(u_0)$, respectively, as the slopes of the straight-line segments representing the edge $u_{\ell-1} u_\ell$ in $\Gamma'_0$ and $\Gamma'_1$ are larger than or equal to those of $\rho_0(u_\ell)$ and $\rho_1(u_\ell)$, respectively, and as the vertical line $x=x^*_j$ lies to the right of all the vertices of $G_{j-1}$ both in $\Gamma'_0$ and in $\Gamma'_1$. The drawings $\Gamma'_0$ and $\Gamma'_1$ are upward, given that the vertices of $G_j$ have the same $y$-coordinates they have in $\Gamma_0$ and $\Gamma_1$, respectively, and given that $\Gamma_0$ and $\Gamma_1$ are upward drawings.
	
	% Clearly, the obtained drawings are planar, as the slopes of the straight-line segments representing edges $u_0 u_1$ and $u_{l-1} u_\ell$ are respectively smaller and larger than those of $\ell_1$ and $\ell_2$, respectively, and the vertical line $x=x^*$ completely lies to the right of the vertices of $G_{j-1}$. Further $\Gamma'_0$ and $\Gamma'_1$ are upward drawings of $G_j$ since the $y$-coordinates of the vertices are the same as in $\Gamma_0$ and $\Gamma_1$, respectively (see \blue{Fig.}~\ref{fig:ear-deco-gj}).
	
	In order to conclude the proof, we show that the obtained drawings, together with $\Gamma_0[G_j]$ and $\Gamma_1[G_j]$, fulfill Conditions~(\ref{pr:one})--(\ref{pr:three}) of the definition of an hvh-pair.
	Since ($\Gamma_0$,$\Gamma'_0$) and ($\Gamma_1$, $\Gamma'_1$) are pairs of upward planar drawings of $G_j$ with the same $y$-assignment, by \autoref{le:st-y-assignment-equivalent}, Conditions (\ref{pr:one}) and (\ref{pr:three}) hold true. In order to prove Condition~(\ref{pr:two}), first recall that by construction $\Gamma'_0$ and $\Gamma'_1$ have the same $x$-assignment. Also, since all the intermediate vertices of any path $P_j$ are drawn on the vertical line $x=x^*_j$ in both $\Gamma'_0$ and $\Gamma'_1$, and the circular ordering of the edges around each vertex is the same in both $\Gamma'_0$ and $\Gamma'_1$, we have that the sequence of vertices and edges crossed by each vertical line in $\Gamma'_0$ and $\Gamma'_1$ is the same, thus implying Condition~(\ref{pr:two}).
\end{proof}

Combining \autoref{le:fast-morph} with \autoref{le:good-reduced} we obtain the following result.

\begin{theorem}\label{th:st-reduced-graphs}
Let $\Gamma_0$ and $\Gamma_1$ be any two upward planar drawings of a reduced plane $st$-graph. There is a $3$-step upward planar morph from $\Gamma_0$ to $\Gamma_1$.
\end{theorem}

%%%%%%%%%
%%%%%%%%%
%%%%%%%%%
%%%%%%%%%
%%%%%%%%%
\subsection{General Plane $st$-Graphs}\label{sse:general-st-graph}
We now turn our attention to general plane $st$-graphs. We restate here, in terms of plane $st$-graphs, a result by Hong and Nagamochi~\cite{DBLP:journals/jda/HongN10} that was originally formulated in terms of hierarchical plane (undirected) graphs.

% \begin{theorem}[\cite{DBLP:journals/jda/HongN10}, Theorem~8]\label{th:hongANDnaga}
% Let $\Gamma$ be an upward planar drawing of an internally $3$-connected upward plane $st$-graph~$G$. There exists a strictly-convex upward planar drawing of $G$ with the same $y$-assignment as~$\Gamma$.
% \end{theorem}

\begin{theorem}[Hong and Nagamochi~\cite{DBLP:journals/jda/HongN10}, Theorem~8]\label{th:hongANDnaga-biconnected}
Consider an internally $3$-connected plane $st$-graph~$G$ and let $y_G$ be a $y$-assignment for the vertices of $G$ such that each vertex $v$ is assigned a value $y_G(v)$ that is greater than those assigned to its predecessors. There exists a strictly-convex upward planar drawing of $G$ satisfying $y_G$.
\end{theorem}

% \begin{theorem}[see Hong and Nagamochi~\cite{DBLP:journals/jda/HongN10}]\label{th:hongANDnaga}
% %	Let $\Gamma$ be an upward planar drawing of an internally $3$-connected upward plane $st$-graph~$G$. There exists a convex upward planar drawing $\Gamma'$ of $G$ with the same $y$-assignment as~$\Gamma$.
% %	
% 	Let $G$ be an upward plane internally $3$-connected graph \red{ and let $y_G$ be a valid $y$-assignment for $G$}. There exists a convex upward planar drawing $\Gamma$ of $G$ satisfying $y_G$.
% \end{theorem}
% In some cases we will use a weaker version of this theorem, that is easy to prove by replacing induced paths with single edges.

% \begin{framed}
% 	\begin{corollary}\label{co:hongANDnaga-subdivision}
% 		Let $\Gamma$ be an upward planar drawing of a subdivision of an internally $3$-connected upward plane $st$-graph~$G$. There exists a weakly-convex upward planar drawing $\Gamma'$ of $G$ with the same $y$-assignment as~$\Gamma$.
% 	\end{corollary}
% \end{framed}
% \hfil $\downarrow$ \hfil

% 	\begin{corollary}\label{co:hongANDnaga-subdivision}
% 	Let $G$ be an upward plane graph that is a subdivision of an internally $3$-connected upward plane graph \red{ and let $y_G$ be a valid $y$-assignment for $G$}. There exists a weakly-convex upward planar drawing $\Gamma$ of $G$ satisfying $y_G$.
% \end{corollary}

We use \autoref{th:hongANDnaga-biconnected} to prove the following lemma, which allows us to restrict our attention to maximal plane $st$-graphs.

\newcommand{\theoremstgraphmaximal}{Let $\Gamma_0$ and $\Gamma_1$ be two upward planar drawings of an $n$-vertex  plane $st$-graph $G$. 
	
Suppose that an algorithm $\mathcal A$ exists that constructs an $f(r)$-step upward planar morph between any two upward planar drawings of an $r$-vertex maximal plane $st$-graph. 

Then there exists an $O(f(n))$-step upward planar morph from $\Gamma_0$ to $\Gamma_1$.}

\begin{lemma}\label{th:st-graphs-maximal}
\theoremstgraphmaximal
\end{lemma}

\begin{proof}
By \autoref{le:biconnected-reduced-augmentation} we can assume that $G$ is biconnected.

We augment $G$ to a maximal plane $st$-graph $G^*$ as follows (refer to \blue{Fig.~\ref{fig:triangulating-a}}). For each internal face $f$ of $G$ we add to $G$: (i) a vertex $v_f$ into $f$, (ii) a directed edge from the source-switch $s_f$ of $f$ to $v_f$, and (iii) directed edges from $v_f$ to every other vertex incident to $f$. We also add a vertex $v^*$ into the outer face of $G$, and add directed edges from $v^*$ to all the vertices incident to the outer face of $G$. The resulting graph $G^*$ is a maximal plane $st$-graph (in particular it is internally-$3$-connected) and contains $O(n)$ vertices. 

\begin{figure}[tb]
	\centering
	\subfloat[\label{fig:triangulating-a}]
	{\includegraphics[page=1,height=.3\textwidth]{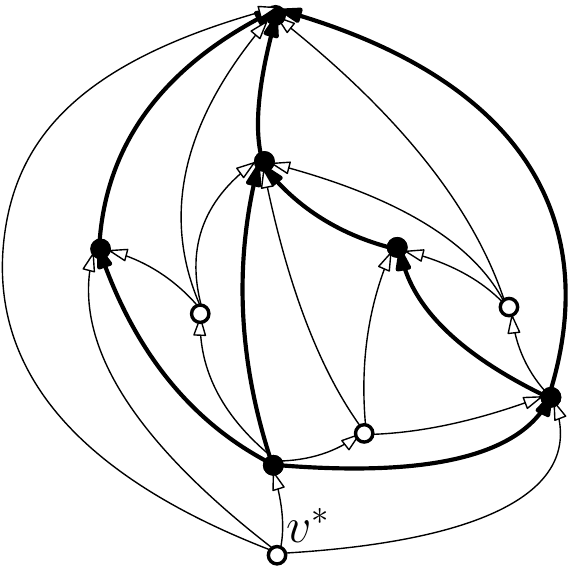}}
	\hfil
	\subfloat[\label{fig:triangulating-b}]
	{\includegraphics[page=2,height=.3\textwidth]{figures/Triangulating}}
	\caption{Illustration for the proof of \autoref{th:st-graphs-maximal}. (a) A biconnected plane $st$-graph $G$ (shown with black vertices and thick edges) and its augmentation to a maximal plane $st$-graph $G^*$ (by the addition of the white vertices and white-arrowed edges). (b) The construction of an upward planar drawing $\Gamma^*_i$ of $G^*$ in which each vertex of $G^*$ that is also in $G$ has the same $y$-coordinate as in $\Gamma_i$.}\label{fig:triangulating}
\end{figure}

Denote by $y^0_G$ the $y$-assignment for the vertices of $G$ that is induced by $\Gamma_0$. We define a $y$-assignment $y^0_{G^*}$ for  the vertices of $G^*$ by setting:

\begin{itemize}
	\item $y^0_{G^*}(v)=y^0_{G}(v)$ for each vertex $v\in V(G)$;
	\item for each vertex $v_f$ of $G^*$ inserted into an internal face $f$ of $G$, a value for $y^0_{G^*}(v_f)$ that is larger than $y^0_{G^*}(s_f)$ and smaller than $y^0_{G^*}(v)$, for every other vertex $v$ incident to $f$; and 
	\item for the vertex $v^*$ of $G^*$ inserted into the outer face of $G$, a value for $y^0_{G^*}(v^*)$ that is smaller than $y^0_{G^*}(v)$, for every vertex $v\neq v^*$ of $G^*$.
\end{itemize} 

We similarly define a $y$-assignment $y^1_{G^*}$ for the vertices of $G^*$ using the $y$-coordinates of $\Gamma_1$.

Note that, for $i=0,1$, each vertex $v$ of $G^*$ has been assigned a value $y^i_{G^*}$ that is greater than those assigned to its predecessors. We can hence use \autoref{th:hongANDnaga-biconnected} to construct upward planar drawings $\Gamma^*_0$ and $\Gamma^*_1$ of $G^*$ satisfying $y^0_{G^*}$ and $y^1_{G^*}$, respectively (refer to \blue{Fig.~\ref{fig:triangulating-b}}).  

By \autoref{le:st-y-assignment-equivalent} we have that the drawings $\Gamma^*_0[G]$ and $\Gamma_0$ are left-to-right equivalent.	Therefore, by \autoref{le:shift}, the linear morph $\mathcal M_0$ from $\Gamma_0$ to $\Gamma^*_0[G]$ is unidirectional and planar. Such a morph is also upward since both $\Gamma_0$ and $\Gamma^*_0[G]$ are upward planar and left-to-right equivalent. Analogously, the linear morph $\mathcal M_1$ from $\Gamma^*_1[G]$ to  $\Gamma_1$ is upward planar. 

We now apply algorithm $\mathcal A$ to construct an $O(n)$-step upward planar morph from $\Gamma^*_0$ to $\Gamma^*_1$ and restrict such a morph to the vertices and edges of $G$ to obtain an $O(n)$-step upward planar morph $\mathcal M_{01}$ from $\Gamma^*_0[G]$ to $\Gamma^*_1[G]$.

An upward planar morph $\mathcal M$ from $\Gamma_0$ to $\Gamma_1$ is finally obtained as the concatenation of $\mathcal M_0$, $\mathcal M_{01}$, and $\mathcal M_1$. The number of steps of $\mathcal M$ is equal to the number of steps of $\mathcal M_{01}$ plus two, hence it is in $O(n)$. This concludes the proof.
\end{proof}

In the following we will present an algorithm that constructs an upward planar morph between two upward planar drawings of a maximal plane $st$-graph. Before doing so, we need to introduce one more tool. The \emph{kernel} of a polygon $P$ is the set of points $p$ inside or on $P$ such that, for any point $q$ on $P$, the open segment $\overline{pq}$ lies inside $P$. We have the following.

\begin{lemma}\label{le:convexify-around-vertex}
Let $\Gamma$ be an upward planar drawing of an internally $3$-connected plane $st$-graph $G$, let $f$ be an internal $st$-face of $G$, and let $P$ be the polygon representing $f$ in $\Gamma$.   

There exists an upward planar drawing $\Gamma'$ of $G$ such that the polygon representing the boundary of $f$ is strictly convex and $\mathcal M = \langle \Gamma, \Gamma' \rangle$ is a unidirectional upward planar morph. 
Further, if $v$ is a vertex incident to $f$ that is in the kernel of $P$ in~$\Gamma$, then $v$ is in the kernel of the polygon representing the boundary of $f$ throughout $\mathcal M$.
\end{lemma}

\begin{proof}
Denote by $y_G$ the $y$-assignment for the vertices of $G$ induced by $\Gamma$.  By \autoref{th:hongANDnaga-biconnected}, there exists a strictly-convex upward planar drawing $\Gamma'$ of $G$ satisfying $y_G$. 
Thus, by~\autoref{le:st-y-assignment-equivalent} and since $G$ is a plane $st$-graph, we have that $\Gamma$ and $\Gamma'$ are left-to-right-equivalent drawings. By~\autoref{le:shift}, the linear morph $\mathcal M$ from $\Gamma$ to $\Gamma'$ is unidirectional and planar. Since $\Gamma$ and $\Gamma'$ are upward and left-to-right equivalent, it follows that $\mathcal M$ is an upward planar morph.

Consider now a vertex $v$ incident to $f$ that is in the kernel of $P$ in~$\Gamma$. Since the polygon representing the boundary of $f$ in $\Gamma'$ is strictly convex, $v$ is also in the kernel of such a polygon. Augment $G$ to a graph $G_*$ by introducing (suitably oriented) edges connecting $v$ to the vertices incident to $f$ that are not already adjacent to~$v$. Draw these edges in $\Gamma$ and $\Gamma'$ as straight-line segments, obtaining two drawings $\Gamma_*$ and $\Gamma_*'$ of $G_*$. Since $v$ is in the kernel of the polygon representing the boundary of $f$ both in $\Gamma$ and in $\Gamma'$, and since $\Gamma$ and $\Gamma'$ are upward planar and left-to-right equivalent, we have that $\Gamma_*$ and $\Gamma_*'$ are left-to-right equivalent upward planar drawings of $G_*$. By the same arguments used for $\mathcal M$, we have that the linear morph $\mathcal M_* = \langle \Gamma_*, \Gamma_*'\rangle$ is planar. Hence, $v$ is in the kernel of the polygon representing the boundary of $f$ throughout $\mathcal M$.
\end{proof}

Given two upward planar straight-line drawings $\Gamma_0$ and $\Gamma_1$ of a maximal plane $st$-graph $G$, our strategy for constructing an upward planar morph from $\Gamma_0$ to $\Gamma_1$ consists of the following steps:
\begin{enumerate}[(1)]
	\item we find a simple vertex $v$ of $G$ of degree at most $5$; 
	\item we remove $v$ and its incident edges from $G$, $\Gamma_0$, and $\Gamma_1$, obtaining upward planar drawings $\Gamma^-_0$ and $\Gamma^-_1$ of an upward plane graph $G^-$; 
	\item we triangulate $G^-$, $\Gamma^-_0$, and $\Gamma^-_1$ by inserting edges incident to a former neighbor $u$ of $v$, obtaining upward planar drawings $\Gamma'_0$ and $\Gamma'_1$ of a maximal plane $st$-graph $G'$; 
	\item we apply induction in order to construct an upward planar morph $\mathcal M'$ from $\Gamma'_0$ to $\Gamma'_1$; and 
	\item we remove from $\mathcal M'$ the edges incident to $u$ that are not in $G$ and insert $v$ and its incident edges in $\mathcal M'$, thus obtaining an upward planar morph from $\Gamma_0$ to $\Gamma_1$.
\end{enumerate} 
	
	 In order for this strategy to work, we need $u$ to satisfy certain properties, which are expressed in the upcoming definition of {\em distinguished neighbor}; further, we need to perform one initial (and one final) upward planar morph so to convexify the polygon representing what will be called a {\em characteristic cycle}. 

Let $v$ be a simple vertex with degree at most $5$ in a maximal plane $st$-graph $G$. Let $G(v)$ be the subgraph of $G$ induced by $v$ and its neighbors. 

A predecessor $u$ of $v$ in $G$ is a \emph{distinguished predecessor} if it satisfies the following properties: (a) for each predecessor $w$ of $v$, there is a directed path in $G(v)$ from $w$ to $v$ through~$u$; (b) $u$ is the only predecessor of $v$ if its degree is $3$; and (c) $v$ has at most two predecessors if its degree is $4$ or $5$. 

A successor $u$ of $v$ in $G$ is a \emph{distinguished successor} if it satisfies the following properties: (a) for each successor $w$ of $v$, there is a directed path in $G(v)$ from $v$ to $w$ through $u$; (b) $u$ is the only successor of $v$ if its degree is $3$; and (c) $v$ has at most two successors if its degree is $4$ or $5$. 

A neighbor of $v$ is a \emph{distinguished neighbor} if it is either a distinguished predecessor or successor of $v$. Examples of distinguished neighbors are in \autoref{fi:distinguished-neighbors}. We are going to exploit the following.

\begin{figure}[t]
	\centering
	\subfloat[]
	{\includegraphics[page=1,height=.2\textwidth]{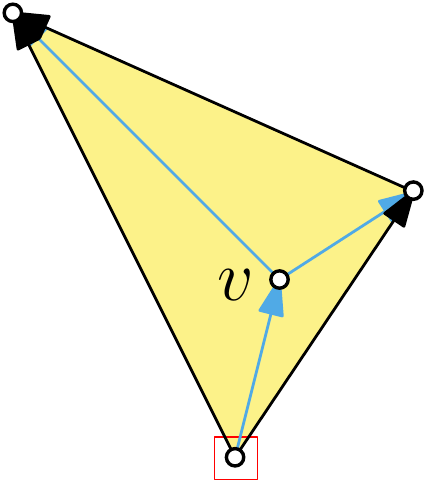}\label{fi:characteristic-a}}
	\hfil
	\subfloat[]
	{\includegraphics[page=3,height=.2\textwidth]{figures/characteristic}\label{fi:characteristic-b}}
	\hfil
	\subfloat[]
	{\includegraphics[page=4,height=.2\textwidth]{figures/characteristic}\label{fi:characteristic-c}}
	\hfil
	\subfloat[]
	{\includegraphics[page=5,height=.2\textwidth]{figures/characteristic}\label{fi:characteristic-d}}
	\hfil
	\subfloat[]
	{\includegraphics[page=2,height=.2\textwidth]{figures/characteristic}\label{fi:characteristic-e}}
	\hfil
	\subfloat[]
	{\includegraphics[page=6,height=.18\textwidth]{figures/characteristic}\label{fi:characteristic-f}}
	\hfil
	\subfloat[]
	{\includegraphics[page=7,height=.2\textwidth]{figures/characteristic}\label{fi:characteristic-g}}
	\hfil
	\subfloat[]
	{\includegraphics[page=8,height=.2\textwidth]{figures/characteristic}\label{fi:characteristic-h}}
	\caption{Distinguished predecessors (enclosed by red squares), distinguished successors (enclosed by red circles), and characteristic cycles (filled yellow). Note that, in (e), (g), and (h), the vertex $s_2$ is not a distinguished successor of $v$; indeed, although for every successor $w$ of $v$ there is a directed path in $G(v)$ from $v$ to $w$ through $s_2$, we have that $v$ has more than two successors.}
	\label{fi:distinguished-neighbors}
\end{figure}

\newcommand{\lemmadistinguishedstatement}{The vertex $v$ has at most one distinguished predecessor, at most one distinguished successor, and at least one distinguished neighbor.}
\begin{lemma}\label{le:distinguished}
	\lemmadistinguishedstatement
\end{lemma}

\begin{proof}
Suppose, for a contradiction, that $v$ has (at least) two distinguished predecessors $u_1$ and $u_2$. Since $u_1$ is a distinguished predecessor of $v$ and $u_2$ is a predecessor of $v$, it follows that $G$ contains a directed path $u_2 \dots u_1 v$; further, since $u_2$ is a distinguished predecessor of $v$ and $u_1$ is a predecessor of $v$, it follows that $G$ contains a directed path $u_1 \dots u_2 v$. The union of these directed paths contains a directed cycle, a contradiction to the fact that $G$ is an $st$-graph. It follows that $v$ has at most one distinguished predecessor. An analogous argument proves that  $v$ has at most one distinguished successor.

Let $P$ and $S$ be the sets of predecessors and successors of $v$ in $G$, respectively. 

If the degree of $v$ is $3$, then either $|P|=1$ or $|S|=1$. In the former case the only predecessor of $v$ is a distinguished predecessor of $v$, while in the latter case the only successor of $v$ is a distinguished successor of $v$.

Assume next that the degree of $v$ is $4$ or $5$. We prove that, if $|P|\leq 2$, then $v$ has at least one distinguished predecessor. If $|P|=1$, then the only predecessor of $v$ is a distinguished predecessor of $v$. Further, if $|P|=2$, then let $s$ and $p$ be the two predecessors of $v$ in $G$. Since $G$ is maximal, it contains either the directed edge $sp$ or the directed edge $ps$. In the former case $p$ is a distinguished predecessor of $v$, while in the latter case $s$ is a distinguished predecessor of $v$. An analogous proof shows that, if $|S|\leq 2$, then $v$ has at least one distinguished successor. This completes the proof, given that $|P|\leq 2$ or $|S|\leq 2$, since the degree of $v$ is at most $5$.   
\end{proof}

We define the \emph{characteristic cycle} $C(v)$ as follows. Let $c_G(v)$ be the subgraph of $G$ induced by the neighbors of $v$. Since $v$ is simple, the underlying graph of $c_G(v)$ is a cycle. 
If $c_G(v)$ is an $st$-cycle, as in \blue{Figs.~\ref{fi:characteristic-a}}, \blue{~\ref{fi:characteristic-b}}, \blue{~\ref{fi:characteristic-c}}, and \blue{~\ref{fi:characteristic-d}}, then $C(v):=c_G(v)$; in particular, this is always the case if $v$ has degree~$3$.
Otherwise, $c_G(v)$ has two sources $s_1$ and $s_2$ and two sinks $t_1$ and $t_2$. Throughout the rest of this section, we always assume that, if $c_G(v)$ has two sources $s_1$ and $s_2$ and two sinks $t_1$ and $t_2$, then $G$ contains the edges $s_1v$ and $vs_2$; indeed, the cases in which $G$ contains the edges $s_2v$ and $vs_1$, or $t_1v$ and $vt_2$, or $t_2v$ and $vt_1$ are analogous. This assumption implies that $v$ has at least three successors, namely $s_2$, $t_1$, and $t_2$, and hence no distinguished successor. Suppose also, w.l.o.g., that $s_1, t_1, s_2$, and $t_2$ appear in this clockwise order along $c_G(v)$. 
If $v$ has degree $4$, as in \blue{Fig.~\ref{fi:characteristic-e}}, then $C(v)$ is composed of the edges $s_1v$, $vs_2$, $s_2t_2$, and $s_1t_2$. 
Otherwise, $v$ has degree $5$, as in \blue{Figs.~\ref{fi:characteristic-f}}, \blue{~\ref{fi:characteristic-g}}, and\blue{~\ref{fi:characteristic-h}}. Let $v_1$ be the distinguished predecessor of $v$. The directed path $P_1 = v_1vs_2$ splits $c_G(v)$ into two paths $P_2$ and $P_3$ with length $2$ and $3$, respectively. Then $C(v)$ is composed of $P_1$ and $P_3$. We have the following structural lemma.

\newcommand{\characteristiccycleproperties}{
The characteristic cycle $C(v)$ is an $st$-cycle which contains all the distinguished neighbors of $v$. Further, all the vertices of $c_G(v)$ not belonging to $C(v)$ are adjacent to all the distinguished neighbors of $v$.
}

\begin{lemma}\label{le:characteristic-cycle-properties}
\characteristiccycleproperties
\end{lemma}

\begin{proof}
If $c_G(v)$ is an $st$-cycle, then by construction $C(v)$ coincides with $c_G(v)$, hence $C(v)$ is an $st$-cycle which contains all the neighbors (and in particular all the distinguished neighbors) of $v$, and there are no vertices of $c_G(v)$ not belonging to $C(v)$. In the following we hence assume that $c_G(v)$ is not an $st$-cycle. 

If $v$ has degree $4$, then by construction $C(v)$ consists of two directed paths, namely $s_1vs_2t_2$ and $s_1t_2$, hence it is an $st$-cycle which contains the only distinguished neighbor (namely $s_1$) of $v$. The only vertex of $c_G(v)$ not belonging to $C(v)$, namely $t_1$, is adjacent to $s_1$.

If $v$ has degree $5$, then observe that $v$ is neither a source nor a sink of $C(v)$, as $C(v)$ contains the directed edges $v_1v$ and $vs_2$; further, $s_2$ is neither a source nor a sink of $C(v)$, as $C(v)$ contains the directed edge $vs_2$ and the directed edge of $P_3$ outgoing $s_2$. Since the underlying graph of $C(v)$ is a cycle with $5$ vertices, it follows that $C(v)$ has one source and one sink, hence it is an $st$-cycle. Further, by construction $C(v)$ contains $v_1$, hence it contains all the distinguished neighbors of $v$. Finally, the only vertex of $c_G(v)$ not belonging to $C(v)$ is the internal vertex of $P_2$, which is adjacent to $v_1$. 
\end{proof}

Characteristic cycles are used in order to prove the following.

\newcommand{\distinguishedkernel}{
Let $\Gamma$ be any upward planar drawing of $G$. There is a unidirectional upward planar morph $\langle \Gamma, \Gamma' \rangle$, where in $\Gamma'$ the distinguished neighbors of $v$ are in the kernel of the polygon representing~$c_G(v)$.
}

\begin{lemma}\label{le:distinguished-kernel}
\distinguishedkernel
\end{lemma}

\begin{proof}
By Lemma~\ref{le:characteristic-cycle-properties} the distinguished neighbors of $v$ belong to $C(v)$. If the polygon $P$ representing $C(v)$ in $\Gamma$ is convex, then each distinguished neighbor of $v$ {\em sees} the other vertices of $C(v)$, meaning that the open straight-line segment connecting any distinguished neighbor of $v$ with any other vertex of $C(v)$ lies inside $P$, and hence inside the polygon representing $c_G(v)$ in $\Gamma$. Again by Lemma~\ref{le:characteristic-cycle-properties} the vertices of $c_G(v)$ that are not in $C(v)$ are adjacent to the distinguished neighbors of $v$, which hence see every vertex of $c_G(v)$. It follows that the distinguished neighbors of $v$ are in the kernel of the polygon representing $c_G(v)$ in $\Gamma$, and we can just define $\Gamma':=\Gamma$; note that no morph is actually needed in order to obtain the desired drawing $\Gamma'$ from $\Gamma$.
	
If $P$ is not convex, then we show how a unidirectional upward planar morph can be employed in order to transform $\Gamma$ into an upward planar drawing $\Gamma'$ in which the polygon representing $C(v)$ is convex, thus bringing us back to the previous case. Let $G^\circ$ be the subgraph of $G$ obtained by removing all the vertices and edges in the interior of $C(v)$ and let $\Gamma^\circ$ be $\Gamma[G^\circ]$. Observe that only $v$ and its incident edges might be removed from $G$ in order to obtain $G^\circ$.
	
We prove that $G^\circ$ is $3$-connected. Suppose, for a contradiction, that $G^\circ$ contains a $2$-cut $\{a,b\}$. If $v$ was removed from $G$ in order to obtain $G^\circ$, then $\{a,b,v\}$ is a $3$-cut of $G$. Since $G$ is maximal, any $3$-cut induces a separating triangle, i.e., a $3$-cycle with vertices both on the inside and on the outside. However, since $v$ is simple, it is not part of any separating triangle, a contradiction. Assume next that $v$ was not removed from $G$ in order to obtain $G^\circ$. If $v \in \{a,b\}$, then $\{a,b\}$ is also a $2$-cut of $G$, contradicting the fact that $G$ is maximal  (and hence $3$-connected). Finally, if $v \notin \{a,b\}$, then $\{a,b,v\}$ is a $3$-cut of $G$, and a contradiction can be derived as in the case in which $v$ was removed from $G$. 

Since $G^\circ$ is $3$-connected and $C(v)$ is an $st$-cycle (by Lemma~\ref{le:characteristic-cycle-properties}), we can  apply Lemma~\ref{le:convexify-around-vertex} to construct an upward planar drawing $\Gamma^\diamondsuit$ of $G^\circ$ such that $C(v)$ is strictly convex in $\Gamma^\diamondsuit$ and $\langle \Gamma, \Gamma^\diamondsuit \rangle$ is a unidirectional upward planar morph. 

We obtain our desired upward planar drawing $\Gamma'$ of $G$ from $\Gamma^\diamondsuit$ as follows. 

If $G^\circ$ contains $v$, then we simply augment $\Gamma^\diamondsuit$ by drawing the edges that are in $G$ but not in $G^\circ$ as straight-line segments, thus obtaining $\Gamma'$. The convexity of $C(v)$ in $\Gamma^\diamondsuit$ implies that no crossings are introduced because of this augmentation. Further, as in the proof of Lemma~\ref{le:convexify-around-vertex}, we have that $\Gamma$ and $\Gamma'$ have the same $y$-assignment, hence by Lemma~\ref{le:st-y-assignment-equivalent} they are left-to-right equivalent, and thus by Lemma~\ref{le:shift} the linear morph $\langle \Gamma, \Gamma' \rangle$ is unidirectional and upward planar.

If $G^\circ$ does not contain $v$, then we need to determine a placement for $v$ in $\Gamma^\diamondsuit$ in order to obtain $\Gamma'$. We insert $v$ in the interior of the convex polygon representing $C(v)$ in $\Gamma^\diamondsuit$, so that its $y$-coordinate is the same as in $\Gamma$. We draw the edges incident to $v$ as straight-line segments. This ensures that $\Gamma$ and $\Gamma'$ have the same $y$-assignment and hence, as in the previous case, that the linear morph between them is unidirectional and upward planar.
\end{proof}

The following concludes our discussion on maximal plane $st$-graph.

\newcommand{\sttriangulatedgraphs}{Let $\Gamma_0$ and $\Gamma_1$ be two upward planar drawings of an $n$-vertex maximal plane $st$-graph $G$. There is an $O(n)$-step upward planar \mbox{morph from $\Gamma_0$ to $\Gamma_1$.}}

\begin{theorem}\label{th:st-triangulated-graphs}
\sttriangulatedgraphs
\end{theorem}

\begin{proof}
The proof is by induction on $n$. In the base case we have $n=3$, hence $\Gamma_0$ and $\Gamma_1$ are two triangles. We show that $\Gamma_0$ and $\Gamma_1$ form an hvh-pair. Denote by $u$ and $w$ the source and the sink of $G$, respectively. Observe that the third vertex of $G$, call it $v$, is on the same side of the edge $uw$ in $\Gamma_0$ and in $\Gamma_1$, as $\Gamma_0$ and in $\Gamma_1$ have the same upward planar embedding; assume that $v$ lies to the right of $uw$, the other case is symmetric. For $i=0,1,$ let $\Gamma'_i$ be a drawing of $G$ such that the $x$-coordinate of $u$ and $w$ is $0$, the $x$-coordinate of $v$ is $1$, and $y$-coordinate of each vertex is the same as in $\Gamma_i$. It is easy to see that $\Gamma'_0$ and $\Gamma'_1$ are upward planar drawings of $G$ and that these drawings, together with $\Gamma_0$ and $\Gamma_1$, satisfy Conditions~(\ref{pr:one})--(\ref{pr:three}) of the definition of an hvh-pair. Thus, by \autoref{le:fast-morph}, there exists a $3$-step upward planar morph from $\Gamma_0$ to $\Gamma_1$.
	
Suppose next that $n > 3$. By Lemma~\ref{le:internal-vertex}, $G$ contains a simple vertex $v$ of degree at most $5$. By Lemma~\ref{le:distinguished}, $v$ has at least one distinguished neighbor, which we denote by $u$. Assume for the remainder of the proof that $u$ is a predecessor of $v$, the case in which it is a successor of $v$ being symmetric. By Lemma~\ref{le:distinguished-kernel}, there exists a unidirectional upward planar morph from $\Gamma_0$ to an upward planar drawing, which we denote again by $\Gamma_0$, in which $u$ lies in the kernel of $c_G(v)$. Analogously, by means of a unidirectional upward planar morph, we can ensure that $u$ lies in the kernel of $c_G(v)$ in $\Gamma_1$.
	
In order to obtain the desired morph from $\Gamma_0$ to $\Gamma_1$ we are going to apply induction. For this sake we define an $(n-1)$-vertex maximal plane $st$-graph $G'$, and two upward planar drawings $\Gamma_0'$ and $\Gamma_1'$ of it. The graph $G'$ is obtained from $G$ by removing $v$  and by inserting a directed edge $uq$ for each successor $q$ of $v$ that is not adjacent to $u$ in $G$. These edges are all added inside $c_G(v)$. Note that, by the definition of distinguished predecessor, either $u$ is the only predecessor of $v$, or $v$ has one predecessor $p$ different from $u$, where $G$ contains the directed edge $pu$. The drawings $\Gamma_0'$ and $\Gamma_1'$ are obtained from $\Gamma_0$ to $\Gamma_1$, respectively, by removing $v$ and its incident edges and by drawing the edges of $G'$ \mbox{not in $G$ as straight-line segments. }

For $i=0,1$, since $u$ lies in the kernel of the polygon representing $c_G(v)$ in $\Gamma_i$, we have that $\Gamma'_i$ is planar. We prove that $\Gamma_i'$ is upward. Every successor $q$ of $v$ has a $y$-coordinate larger than the one of $v$ in $\Gamma_i$; since $u$ has a $y$-coordinate smaller than the one of $v$ in $\Gamma_i$, it follows that the edge from $u$ to $q$ is monotonically increasing in the $y$-direction in $\Gamma'_i$. Since all the edges of $G'$ that are also in $G$ are drawn as in $\Gamma_i$ and since $\Gamma_i$ is an upward drawing, it follows that $\Gamma'_i$ is an upward planar drawing of $G'$. Observe that $G'$ is an $st$-graph; indeed, it suffices to note that the edges that are removed from $G$ do not result in any new source or sink in $G'$: (i) no successor $q$ of $v$ becomes a source in $G'$, as a directed edge $uq$ is inserted in $G'$ if it is not in $G$; (ii) no predecessor $p$ of $v$ different from $u$, if any, becomes a sink in $G'$, as the directed edge $pu$ belongs to $G$; and (iii) $u$ does not become a sink in $G'$ as $v$ has at least one successor in $G$. Finally, note that $G'$ is maximal, since $G$ is maximal and the edges added to $G'$ triangulate the interior of $c_G(v)$. It follows that $G'$ is a maximal plane $st$-graph.

By induction, there is an upward planar morph $\mathcal M'=\langle \Gamma'_0 = \Lambda_0, \Lambda_1, \dots, \Lambda_k=\Gamma'_1\rangle$ from $\Gamma'_0$ to $\Gamma'_1$. In the following we transform $\mathcal M'$ into an upward planar morph $\mathcal M$ between two upward planar drawings $\Delta_0$ and $\Delta_1$ of $G$. This will be done by inserting $v$ at a suitable point in the drawing of $G'$ at any time instant of the morph $\mathcal M'$ and by drawing the edges incident to $v$ as straight-line segments. We will later show that $\mathcal M$ is actually composed of $k$ linear morphs.\footnote{This insertion problem has been studied and solved in~\cite{DBLP:journals/siamcomp/AlamdariABCLBFH17} for planar morphs of undirected graphs. Here we cannot immediately reuse the results in~\cite{DBLP:journals/siamcomp/AlamdariABCLBFH17}, as we need to preserve the upwardness of the drawing throughout the morph. However, the property that every drawing of $G'$ in $\mathcal M'$ is upward significantly simplifies the problem of inserting $v$ in $\mathcal M'$ so to obtain an upward planar morph of $G$.} 

Let $\varepsilon>0$ be a sufficiently small value such that the following properties are satisfied throughout $\mathcal M'$:

\begin{enumerate}[(a)]
	\item for each successor $q$ of $v$ in $G$, it holds true that $y(q) > y(u) + \varepsilon$; 
	\item if $v$ has a predecessor $p\neq u$ in $G$, then $y(p) < y(u) - \varepsilon$; and
	\item for any segment $s$ of $c_G(v)$ not incident to $u$, the line through $s$ does not intersect the disk $\delta$ with radius $\varepsilon$ centered at $u$.
\end{enumerate}

Since $\mathcal M'$ is an upward planar morph and since $G'$ contains edges from $u$ to every successor $q$ of $v$ and from every predecessor $p$ of $v$ to $u$, it follows that such a value $\varepsilon$ exists; in particular, standard continuity arguments, like the ones used in the proof of F\'ary's Theorem~\cite{Fary}, ensure that Property~(c) is satisfied for a sufficiently small value $\varepsilon>0$.

%For each $i=0, \dots, k$, let $P_i$ the polygon representing $c_G(v)$ in $\Lambda_i$ and let $\varepsilon_i > 0$ be a sufficiently small value such that the following are satisfied in $\Lambda_i$: 
%	\begin{enumerate}[(a)]
%		\item the disk with radius $\varepsilon_i$ centered at $u$ lies inside the kernel of $P_i$; and
%		\item for each successor $q$ of $v$ in $G$, it holds true that $y(q) > y(u) + \varepsilon_i$; and 
%		\item for each predecessor $p\neq u$ of $v$ in $G$, it holds true that  $y(p) < y(u) - \varepsilon_i$.
%	\end{enumerate}

%The existence of such a value $\varepsilon_i$ comes from a standard continuity argument, like the one used in the proof of F\'ary's Theorem~\cite{Fary}, and from the fact that $\Lambda_i$ is an upward planar drawing of $G'$, for $i=0, \dots, k$. Let $\varepsilon = \min_{i=0}^k \varepsilon_i$.  Let $d_i$ be the disk with radius $\varepsilon$ centered at $u$ in $\Lambda_i$.

%For $i=0, \dots, k$, consider the drawing $\Lambda_i$ and denote by $P_i$ the polygon representing $c_G(v)$, by $d_i$ the disk with radius $\varepsilon$ centered at $u$, and by $I_i$ the intersection between $d_i$ and the interior of $P_i$. We show how to obtain an upward planar drawing $\Delta_i$ of $G$ from $\Lambda_i$ by placing $v$ at a point $v_i$ in $I_i$. 
	
We distinguish the cases in which $v$ has degree $3$ or greater than $3$ in $G$.
	
If $v$ has degree $3$ in $G$, then let $a$, $b$, and $c$ be the neighbors of $v$ in $G$, where $a=u$. We choose three values $\alpha, \beta$, and $\gamma$, as discussed below, and then place $v$ at the point $\alpha \cdot a +   \beta  \cdot b + \gamma  \cdot c$ at any time instant of $\mathcal M'$ ($a$, $b$, and $c$ here represent the points at which the corresponding vertices are placed at any time instant of $\mathcal M'$). We choose $\alpha, \beta$, and $\gamma$ as any positive values such that $\alpha + \beta + \gamma = 1$ and such that the point $\alpha \cdot a +   \beta  \cdot b + \gamma  \cdot c$ lies in $\delta$ throughout $\mathcal M'$. Note that $v$ lies inside the triangle $c_G(v)$ for any positive values of $\alpha, \beta$, and $\gamma$ such that $\alpha + \beta + \gamma = 1$ (indeed, the position of $v$ is a convex combination of the ones of $a$, $b$, and $c$); further, choosing $\alpha$ sufficiently close to $1$ ensures that $v$ is at distance at most $\varepsilon$ from $u$, and hence lies inside $\delta$, throughout $\mathcal M'$.
	
Suppose now that $v$ has degree $4$ or $5$ in $G$. Since $u$ is a distinguished predecessor of $v$, we have that $v$ has at most two predecessors in $G$, one of which is $u$. If $v$ has no predecessor other than $u$, then $v$ has at least three successors in $G$; let $w$ be a successor of $v$ not adjacent to $u$ in $G$. If $v$ has a predecessor $p$ different from $u$, then $v$ has at least two successors in $G$; let $w$ be the one adjacent to $p$ in $G$. Note that, in both cases, the directed edge $uw$ belongs to $G'$ but not to $G$ and connects $u$ with a successor $w$ of $v$. We compute a value $\lambda$, as discussed below, and then place $v$ at the point $\lambda \cdot u + (1-\lambda) \cdot w$ at any time instant of $\mathcal M'$ ($u$ and $w$ here represent the points at which the corresponding vertices are placed at any time instant of $\mathcal M'$). We choose $\lambda$ as any positive value smaller than $1$ such that the point $\lambda \cdot u + (1-\lambda) \cdot w$ lies in $\delta$ throughout $\mathcal M'$. Note that $v$ is on the straight-line segment representing the edge $uw$ for any positive value of $\lambda$ smaller than $1$ (indeed, the position of $v$ is a convex combination of the ones of $u$ and $w$); further, choosing $\lambda$ sufficiently close to $1$ ensures that $v$ is at distance at most $\varepsilon$ from $u$, and hence lies inside $\delta$, throughout $\mathcal M'$.

In both cases, the choice of $\varepsilon$ ensures that at any time instant of $\mathcal M$ the drawing of $G$ is upward planar. In particular, Properties~(a) and~(b), together with the fact that every drawing of $G'$ in $\mathcal M'$ is upward, directly ensure that every drawing of $G$ in $\mathcal M$ is upward. Further, Property~(c), together with the fact that every drawing of $G'$ in $\mathcal M'$ is planar, ensures that every drawing of $G$ in $\mathcal M$ is planar. In particular, every point of $\delta$ sees every point of any straight-line segment $s$ of $c_G(v)$ not incident to $u$ in the interior of the polygon representing $c_G(v)$; hence the directed edges from $v$ to its successors cause no crossings throughout $\mathcal M$. Further, if $v$ has a predecessor $p$ different from $u$, the fact that $v$ lies on the straight-line segment connecting $u$ with the neighbor of $p$ in $G$ ensures that $p$ sees $v$ in the interior of the polygon representing $c_G(v)$ throughout $\mathcal M$.

We now prove that $\mathcal M$ consists of $k$ morphing steps. Assume that the degree of $v$ is $3$, the discussion for the case in which the degree of $v$ is $4$ or $5$ being analogous and simpler. Denote by $\Delta_i$ the drawing of $G$ obtained from $\Lambda_i$ by placing $v$ at the point $\alpha \cdot a_i +   \beta  \cdot b_i + \gamma  \cdot c_i$, as discussed above, where by $a_i$, $b_i$, and $c_i$ denote the positions of the vertices $a$, $b$, and $c$ in $\Lambda_i$, respectively. Hence, at any time $t\in [0;1]$ of the linear morph $\langle \Delta_i,\Delta_{i+1}\rangle$, the position of $v$ is 

\begin{eqnarray*}
&& (1-t)\cdot(\alpha \cdot a_i +   \beta  \cdot b_i + \gamma  \cdot c_i) + t\cdot(\alpha \cdot a_{i+1} +   \beta  \cdot b_{i+1} + \gamma  \cdot c_{i+1})\\
&=&  \alpha((1-t)\cdot a_i + t \cdot a_{i+1})+ \beta((1-t)\cdot b_i + t \cdot b_{i+1})+ \gamma((1-t)\cdot c_i + t \cdot c_{i+1}).
\end{eqnarray*}

Hence, the position of $v$ at any time instant of the linear morph $\langle \Delta_i,\Delta_{i+1}\rangle$ is given by the convex combination with coefficients $\alpha$, $\beta$, and $\gamma$ of the positions of $a$, $b$, and $c$. It follows that the upward planar morph $\mathcal M$ defined above coincides with the $k$-step morph $\langle \Delta_0, \Delta_1, \dots, \Delta_k\rangle$. 

Finally, denote by $f(n)$ the number of morphing steps of the described algorithm. We have $f(3)=3$ and $f(n)=4+f(n-1)$, if $n>3$. Indeed, in the inductive case the upward planar morph from $\Gamma_0$ to $\Gamma_1$ consists of:

\begin{itemize}
	\item a first morphing step from the given drawing $\Gamma_0$ to the drawing $\Gamma_0$ in which $u$ lies in the kernel of $c_G(v)$;
	\item a second morphing step $\langle \Gamma_0,\Delta_0\rangle$, where in $\Gamma_0$ the vertex $u$ lies in the kernel of $c_G(v)$ (note that only $v$ moves during this morphing step);
	\item the morph $\mathcal M=\langle \Delta_0, \Delta_1, \dots, \Delta_k\rangle$, whose number $k$ of steps is the same as in $\mathcal M'$, which is the inductively constructed morph of the $(n-1)$-vertex maximal plane $st$-graph $G'$; hence $k=f(n-1)$;
	\item a second to last morphing step $\langle \Delta_k,\Gamma_1\rangle$, where in $\Gamma_1$ the vertex $u$ lies in the kernel of $c_G(v)$ (note that only $v$ moves during this morphing step); and 
	\item a final morphing step from the drawing $\Gamma_1$ in which $u$ lies in the kernel of $c_G(v)$ to  the given drawing $\Gamma_1$. 
\end{itemize}

The recurrence equation for $f(n)$ solves to $f(n)=4n-9$. This concludes the proof of the theorem.
\end{proof}

\newcommand{\stmain}{Let $\Gamma_0$ and $\Gamma_1$ be two upward planar drawings of an $n$-vertex plane $st$-graph. There exists an $O(n)$-step upward planar morph from $\Gamma_0$ to $\Gamma_1$.}

We finally get the following.

\begin{corollary}\label{co:st-main}
	\stmain
\end{corollary}

\begin{proof}
The statement follows by \autoref{th:st-graphs-maximal} and \autoref{th:st-triangulated-graphs}.
\end{proof}

%%%%%%%%%%%%%%%%%%%%%%%%%%%%%%%%%%%%%%%%%%%%%%%%%%%%%%%%%%%%%%%%%%%%%%%%%%%%%%%
%%%%%%%%%%%%%%%%%%%%%%%%%%%%%%%%%%%%%%%%%%%%%%%%%%%%%%%%%%%%%%%%%%%%%%%%%%%%%%%
%%%%%%%%%%%%%%%%%%%%%%%%%%%%%%%%%%%%%%%%%%%%%%%%%%%%%%%%%%%%%%%%%%%%%%%%%%%%%%%
% ____  _
%|  _ \| | __ _ _ __   ___
%| |_) | |/ _` | '_ \ / _ \
%|  __/| | (_| | | | |  __/
%|_|   |_|\__,_|_| |_|\___|
%
%%%%%%%%%%%%%%%%%%%%%%%%%%%%%%%%%%%%%%%%%%%%%%%%%%%%%%%%%%%%%%%%%%%%%%%%%%%%%%%%%%%%%%%%%%%%%%%%%%%%%%%%%%%%%%%%%%%%%%%%%%%%%%%%%%%%%%%%%%%%%%%%%%%%%%%%%%%%%%%%%%%%%%%%%%%%%%%%%%%%%%%%%%%%%%%%%%%%%%%%%%%%%%%%%%%%%%%%%%%%%%%%%%%%%%%%%%%%%

\section{Upward Plane Graphs}\label{se:general-graphs}

In this section we deal with general upward plane graphs. In order to morph two upward planar drawings of an upward plane graph $G$ we are going to augment the upward planar drawings of $G$ to two upward planar drawings of a (possibly reduced) plane $st$-graph $G'$ and then to use the results of Section~\ref{se:st-graphs} for morphing the obtained upward planar drawings of $G'$. The augmentation process itself uses upward planar morphs. In the following we formally describe this strategy.

Let $G$ be an upward plane graph whose underlying graph is biconnected, let $f$ be a face of $G$ which is not an $st$-face, and let $u$, $v$, and $w$ be three clockwise consecutive switches of $f$. Further, let $v^-$ and $v^+$ be the vertices preceding and succeeding $v$ in clockwise order along the boundary of $f$, respectively, and let $u^-$ and $u^+$ be the vertices preceding  and succeeding $u$ in clockwise order along the boundary of $f$, respectively. We say that $[u, v, w]$ is a \emph{pocket} for $f$ if $\angle (v^-,v,v^+)= \texttt{small}$ and $\angle (u^-,u,u^+) = \texttt{large}$. The following is well-known.

\begin{lemma}[Bertolazzi et al.~\cite{DBLP:journals/algorithmica/BertolazziBLM94}]\label{lem:small-small-large}
Let $G$ be an upward plane graph whose underlying graph is biconnected and let $f$ be a face of $G$ that is not an $st$-face. Then, there exists a pocket $[u, v, w]$ for $f$.
\end{lemma}

Next, we give a lemma that shows how to ``simplify'' a face of an upward plane graph that is not an $st$-graph, by removing one of its pockets.

\newcommand{\morphingbetweenaugmentations}{
Let $G$ be an $n$-vertex (reduced) upward plane graph whose underlying graph is biconnected, let $f$ be a face of $G$ that is not an $st$-face, let $[u, v, w]$ be a pocket for $f$, and let $\Gamma$ be an upward planar drawing of $G$. 

Suppose that an algorithm $\mathcal A$ ($\mathcal A_R$) exists that constructs an $f(r)$-step ($f_R(r)$-step) upward planar morph between any two upward planar drawings of an $r$-vertex (reduced) plane $st$-graph. 

Then, there exists an $O(f(n))$-step (an $O(f_R(n))$-step) upward planar morph from $\Gamma$ to an upward planar drawing $\Gamma^*$ of $G$ in which $u$ sees $w$ inside $f$ and in which $u$ lies below $w$, if the directed path between $u$ and $v$ along the boundary of $f$ is directed from $v$ to $u$, or $u$ lies above $w$, otherwise.
}

\begin{lemma}\label{lem:morphing-between-augmentations}
\red{\morphingbetweenaugmentations}
\end{lemma}

\begin{proof}
	Suppose that the directed path $p_{vu}$ between $u$ and $v$ along the boundary of $f$ is directed from $v$ to $u$ (refer to \blue{Fig.}~\ref{fig:augmentation-a}); the case in which it is directed from $u$ to $v$ can be treated symmetrically.

\begin{figure}[tb]
	\centering
	\subfloat[\label{fig:augmentation-a}]
	{\includegraphics[page=1,width=.33\textwidth]{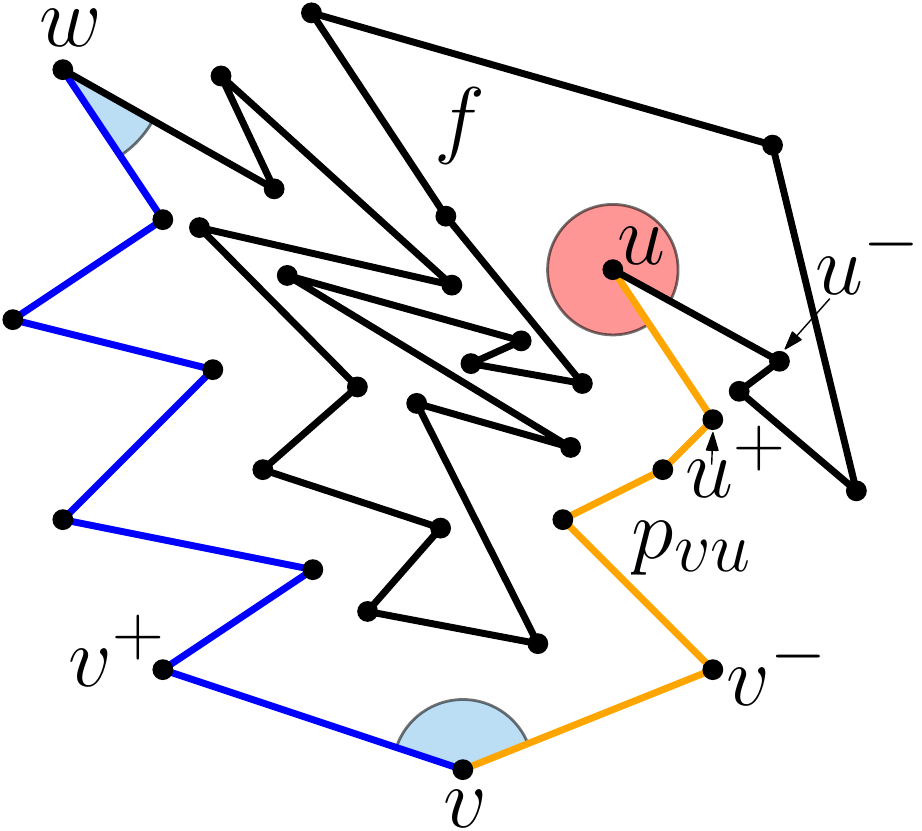}}
	\hfil
	\subfloat[]
	{\includegraphics[page=1,width=.33\textwidth]{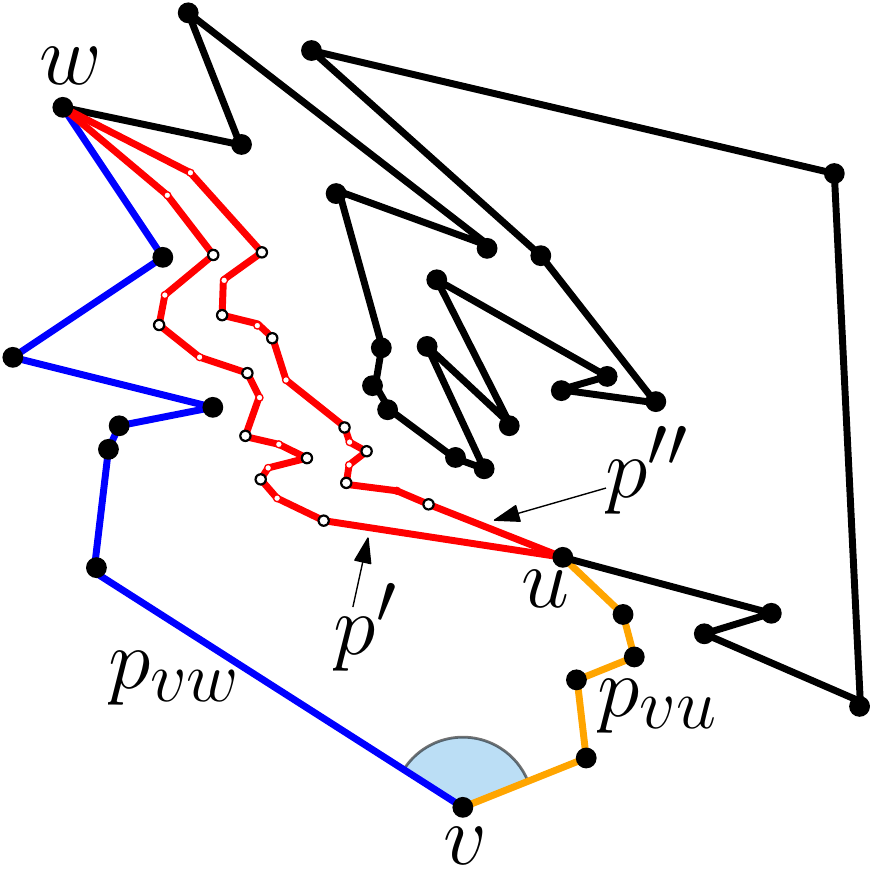}\label{fig:augmentation-b}}
	\hfil
	\subfloat[]
	{\includegraphics[page=1,width=.33\textwidth]{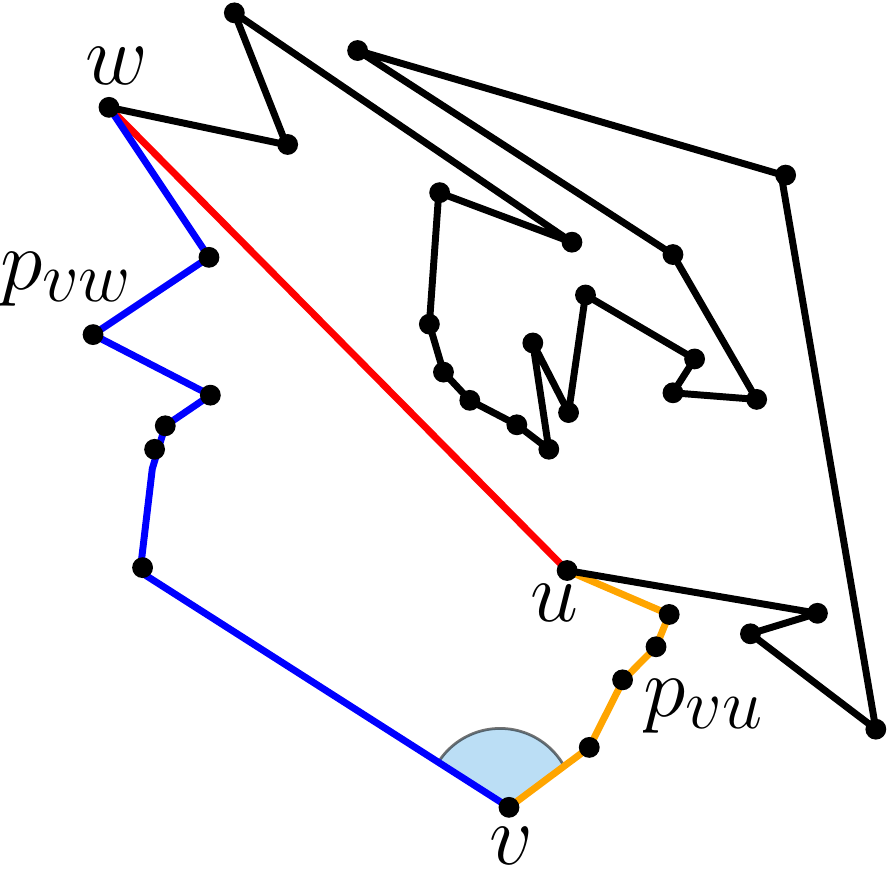}\label{fig:augmentation-c}}
	\caption{Illustrations for the proof of  \autoref{lem:morphing-between-augmentations}. (a) The upward planar drawing $\Gamma$ of $G$; in particular, the illustration shows the face $f$ whose boundary contains the pocket $[u, v, w]$. (b) The upward planar drawing $\Gamma'$ of $G$ (plus the directed paths $p'$ and $p''$). (c) The upward planar drawing $\Gamma^*$ of $G$ (plus the directed edge $uw$). }\label{fig:augmentations}
\end{figure}

	The proof is structured as follows.
	First, we show that there exists an upward planar drawing $\Gamma'$ of~$G$ such that (see \blue{Fig.}~\ref{fig:augmentation-b}): 
	\begin{enumerate}[(i)]
		\item \label{prop:i} the upward planar drawings of two directed paths $p'$ and $p''$ from $u$ to $w$ can be inserted in $\Gamma'$ in the interior of $f$; and 
		\item \label{prop:ii} there exists an $O(f(n))$-step (an $O(f_R(n))$-step) upward planar morph $\mathcal M'$ from $\Gamma$ to $\Gamma'$.
	\end{enumerate}

	Second, we show that there exists an upward planar drawing $\Gamma^*$ of $G$ such that (see \blue{Fig.}~\ref{fig:augmentation-c}):
	\begin{enumerate}[(i)]\setcounter{enumi}{2}
		\item\label{prop:iii} $u$ sees $w$ inside $f$ and $u$ lies below $w$, and
		\item\label{prop:iv}there exists an $O(f(n))$-step (an $O(f_R(n))$-step) upward planar morph $\mathcal M^*$ from $\Gamma'$ to $\Gamma^*$. 
	\end{enumerate}
	The lemma follows from the existence of the drawings $\Gamma'$ and $\Gamma^*$ above, since composing $\mathcal M'$ with $\mathcal M^*$ yields the desired upward planar morph with $O(f(n))$ steps (with $O(f_R(n))$ steps).
	
	The drawing $\Gamma'$ is constructed in four phases. 
	
	In {\bf phase 1} we augment $\Gamma$ to an upward planar drawing $\Gamma_1$ of an upward plane graph $G_1$. Refer to \blue{Fig.}~\ref{fig:preliminary-augmentation-a}. Let $p_{vu}=v u_1 \dots u_k u$ and $p_{vw}=v w_1 \dots w_h w$ be the directed paths from $v$ to $u$ and from $v$ to $w$, respectively, that belong to the boundary of $f$. In order to construct $\Gamma_1$ and $G_1$ we insert, for some sufficiently small $\epsilon>0$, the following directed paths inside $f$ into $\Gamma$ and $G$: 

\begin{figure}[tb]
	\centering
	\subfloat[\label{fig:preliminary-augmentation-a}]
	{\includegraphics[page=1,width=.33\textwidth]{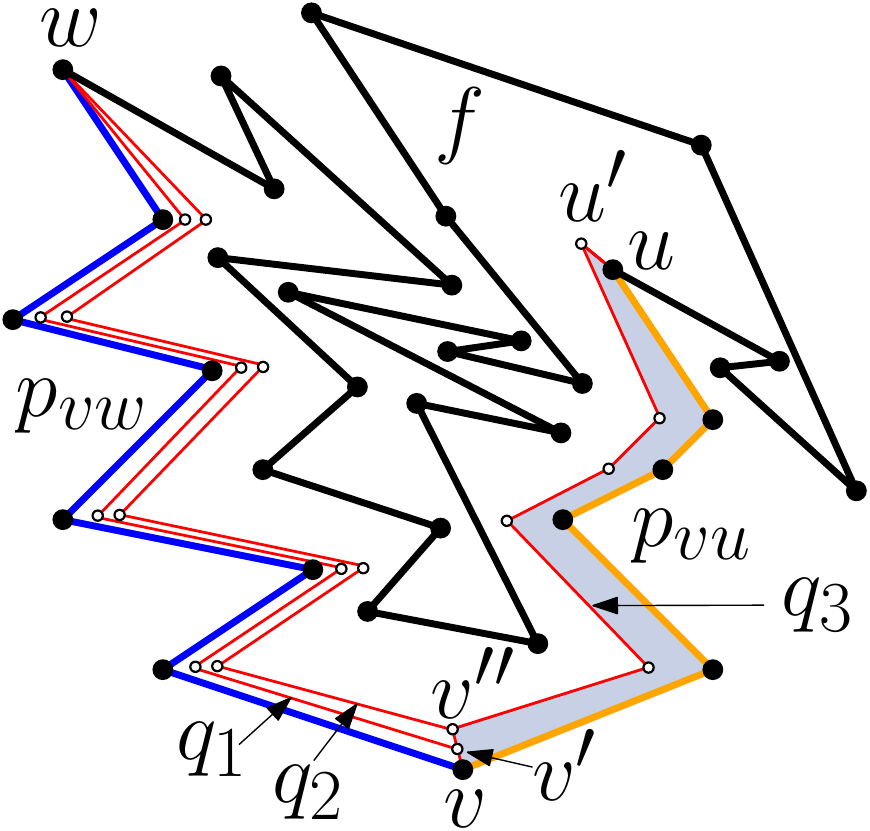}}
	\hfil
	\subfloat[]
	{\includegraphics[page=1,width=.33\textwidth]{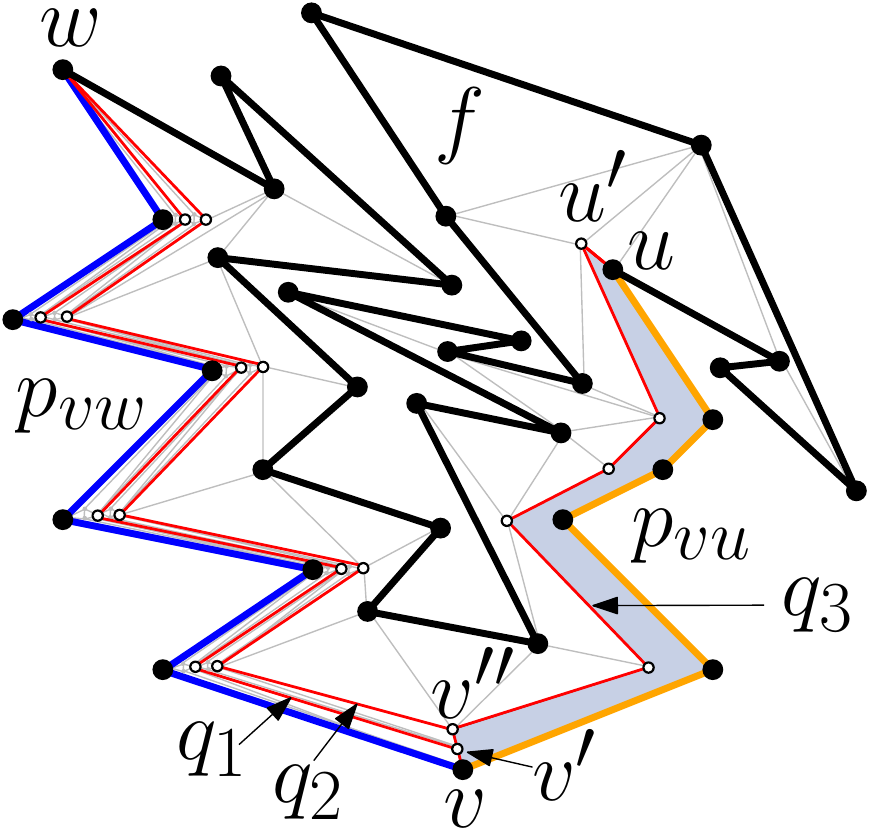}\label{fig:preliminary-augmentation-b}}
	\hfil
	\subfloat[]
	{\includegraphics[page=1,width=.33\textwidth]{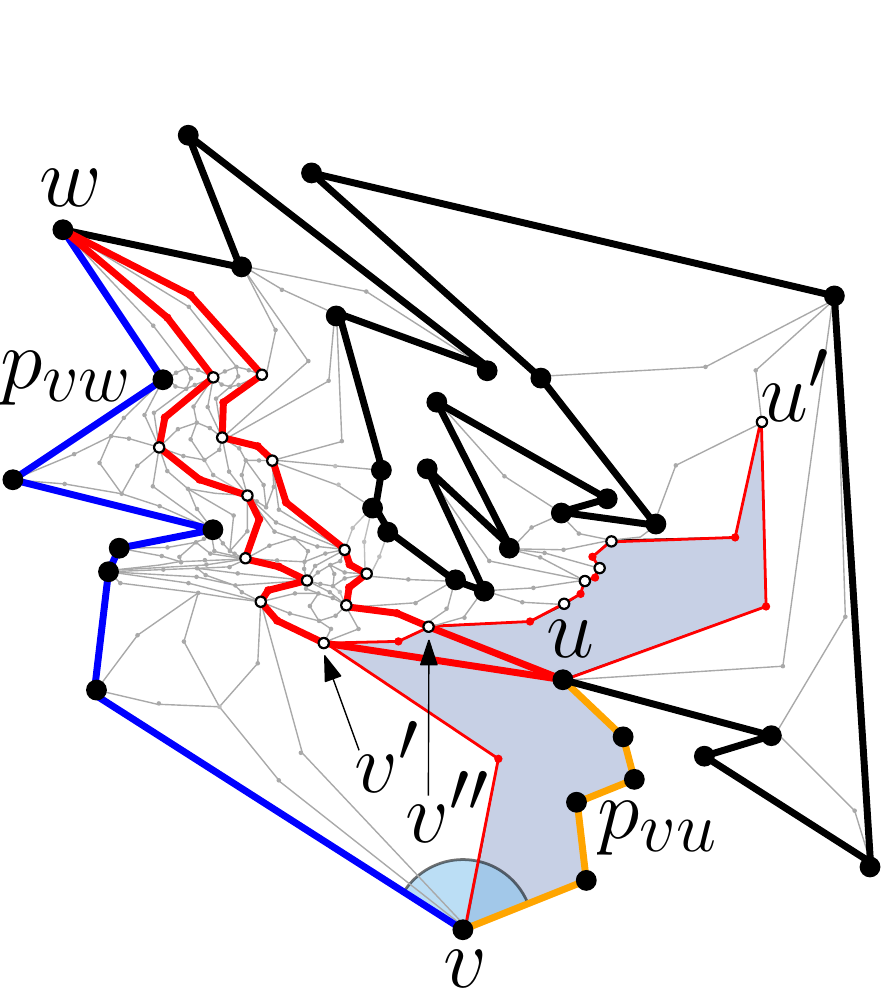}\label{fig:preliminary-augmentation-c}}
	\caption{Construction of the drawing $\Gamma'$; only what happens to the face $f$ is shown. (a) The drawing $\Gamma_1$ of $G_1$. The face $g$ is gray. (b) The drawing $\Gamma_2$ of $G_2$. (c) The drawing $\Gamma_4$ of $G_4$; the paths $p'$ and $p''$ are thick and red.}\label{fig:preliminary-augmentations}
\end{figure}

	\begin{enumerate}[(a)]
		\item a directed path $q_1=v v'w'_1\dots w'_h w$, where $w'_i$ is on the same horizontal line as $w_i$, at horizontal distance $\epsilon/2$ from it, while $v'$ is above $v$, at distance $\epsilon/2$ from it, along the bisector of the angle $\angle (v^-,v,v^+)$; 
		\item a directed path $q_2=v' v''w''_1\dots w''_h w$, where $w''_i$ is on the same horizontal line as $w_i$, at horizontal distance $\epsilon$ from it, while $v''$ is above $v$ and $v'$, at distance $\epsilon$ from $v$, along the bisector of the angle $\angle (v^-,v,v^+)$;   
		\item a directed path $q_3=v''u'_1\dots u'_k u'$, where $u'_i$ is on the same horizontal line as $u_i$, at horizontal distance $\epsilon$ from it, while $u'$ is above $u$, at distance $\epsilon$ from it, along the bisector of the angle $\angle (u^-,u,u^+)$; and  
		\item a directed edge $uu'$.
	\end{enumerate}
	
It is easy to see that the resulting drawing $\Gamma_1$ has no crossings and that the directed paths $q_1$, $q_2$, $q_3$, and $uu'$ are upwardly drawn inside $f$, provided that $\epsilon$ is small enough. Note that there is an $st$-face $g$ of $G_1$ that is delimited by the directed path composed of $vv'v''$ and of $q_3$ and by the directed path composed of $p_{vu}$ and of $uu'$.
	
%	Consider a vertex $z$ incident to $f$ and let $z_1$ and $z_2$ be the neighbors of $z$ preceding and succeeding $z$ in clockwise order along the boundary of $f$, respectively. Denote by $bis(z_1,z,z_2)$ the half-line originating at $z$ that bisects angle $\angle (z_1,z,z_2)$. For each vertex $r \in V(p_{vw}), r \neq u$, we place vertices $r'$ and $r''$ along $bis(r_1,r,r_2)$ at distance $\epsilon/2$ and $\epsilon$ from $r$ (where $r_1$ ($r_2$) is the neighbor of $r$ clockwise preceding (succeeding) $r$ along $f$), and for each vertex $q \in V(p_{vu}), q \neq v$, we place a vertex $q'$ along $bis(q_1,q,q_2)$ at distance $\epsilon$ from $q$
%	(where $q_1$ ($q_2$) is the neighbor of $q$ clockwise preceding (succeeding) $q$ along $f$)
%	, where $\epsilon>0$ is an arbitrarily small value such 
%	that the edges of the paths $p_1=(v',w'_1,\dots,w'_k,w)$ directed from $v'$ to $w$,  $p_2=(v'',w''_1,\dots,w''_k,w)$ directed from $v''$ to $w$, and $p_3=(v'',u'_1,\dots,u'_j,u')$ directed  from $v''$ to $u'$ can be drawn as straight-line segments without introducing any crossings.
%	Finally, we draw directed edges $v v'$, $v' v''$, and $u u'$ as straight-line segments (which does not introduce any crossings). 	
	
	In {\bf phase 2} we augment $\Gamma_1$ to an upward planar drawing $\Gamma_2$ of a plane $st$-graph $G_2$. Refer to \blue{Fig.}~\ref{fig:preliminary-augmentation-b}. In order to construct $\Gamma_2$, we insert edges drawn as straight-line segments into every face of $\Gamma_1$, except for $g$, until no further edge can be inserted while maintaining planarity; each inserted edge is oriented from the endpoint with the lowest $y$-coordinate to the endpoint with the highest $y$-coordinate in $\Gamma_1$ (if the end-points of an edge have the same $y$-coordinate, then we insert two new adjacent vertices, slightly above and below the middle point of that edge, and then keep on inserting edges). This concludes the construction of the drawing $\Gamma_2$ of $G_2$. Since $\Gamma_1$ is upward and planar, it follows that $\Gamma_2$ is upward and planar as well. Further, $G_2$ is an $st$-graph by Lemma~\ref{le:st-faces-iff-st-graph}, since $g$ is an $st$-face, as argued above, and since all the faces of $G_2$ different from $g$ are also delimited by $st$-cycles, as otherwise more edges could have been introduced while maintaining the planarity of $\Gamma_2$; note that every internal face of $\Gamma_2$ different from $g$ is delimited by an upwardly drawn $3$-cycle, while more than $3$ vertices might be incident to the outer face of $\Gamma_2$. 
	
	In {\bf phase 3} we replace each directed edge $uv$ of $G_2$ that does not belong to $G$ (and has been inserted in phase 1 or 2) with a directed path $(u,w_{uv},v)$ and insert $w_{uv}$ at an arbitrary internal point of the edge $uv$ in $\Gamma_2$. Clearly, the resulting graph $G_3$ is a plane $st$-graph and it is reduced if $G$ is. Further, the resulting drawing $\Gamma_3$ is an upward planar drawing of $G_3$. 
	
	In {\bf phase 4} we augment $G_3$ to a plane $st$-graph $G_4$ by adding two directed edges $uv'$ and $uv''$ inside $g$. Observe that $G_4$ is a plane $st$-graph, by Lemma~\ref{le:st-faces-iff-st-graph}, since the directed edges $uv'$ and $uv''$ split $g$ into three $st$-faces. Further, $G_4$ is reduced if $G$ is. Let $p'$ be the directed path composed of $uv'$ and of the (subdivided) directed path $q_1$ and let $p''$ be the directed path composed of $uv'$ and of the (subdivided) directed path $q_2$. Observe that the $st$-cycle $\mathcal D$ of $G_4$ composed of $p'$ and $p''$ does not enclose any vertex of $G$, although it encloses vertices of $G_4$ not in $G$. We construct an upward planar straight-line drawing $\Gamma_4$ of $G_4$ by means of, e.g., the algorithm by Di Battista and Tamassia~\cite{DBLP:journals/tcs/BattistaT88}. 
	
	Now let $\Gamma' = \Gamma_4[G]$. Property~(i) then follows from the fact that upward planar drawings of the directed paths $p'$ and $p''$ from $u$ to $w$ can be inserted in $\Gamma'$ as they are drawn in $\Gamma_4$. Further, since $G_3$ has $O(n)$ vertices, by applying algorithm $\mathcal A$ ($\mathcal A_R$) we can construct an $O(f(n))$-step (an $O(f_R(n))$-step) upward planar morph $\mathcal M_{3,4}$ from $\Gamma_3$ to $\Gamma_4[G_3]$. Since $\Gamma=\Gamma_3[G]$ and $\Gamma' = \Gamma_4[G]$, the restriction of $\mathcal M_{3,4}$ to the vertices and edges of $G$ provides an $O(f(n))$-step (an $O(f_R(n))$-step) upward planar morph $\mathcal M'$ from $\Gamma$ to $\Gamma'$, which proves Property~(ii).
	
	The drawing $\Gamma^*$ is constructed as follows. 
	
	First, we remove from $\Gamma_4$ and $G_4$ all the vertices and edges enclosed by $\mathcal D$. Let $\Gamma_5$ and $G_5$ be the resulting drawing and the resulting graph, respectively. We have that $G_5$ is a plane $st$-graph by Lemma~\ref{le:st-faces-iff-st-graph}; indeed, the face $d$ of $G_5$ delimited by $\mathcal D$ is an $st$-face, as $\mathcal D$ is composed of the directed paths $p'$ and $p''$, and every other face of $G_5$ is an $st$-face since it is also a face of $G_4$, which is a plane $st$-graph. Moreover, $\Gamma_5$ is an upward planar drawing of $G_5$, given that $\Gamma_4$ is an upward planar drawing of $G_4$. 
	
	Second, we augment $G_5$ to a plane $st$-graph $G_6$ by inserting the directed edge $uw$ inside $d$. We construct an upward planar straight-line drawing $\Gamma_6$ of $G_6$ by means of, e.g., the algorithm by Di Battista and Tamassia~\cite{DBLP:journals/tcs/BattistaT88}. 
	
	Now let $\Gamma^* = \Gamma_6[G]$. Property~(iii) then follows from the fact that the directed edge $uw$ lies inside $f$ and is upwardly drawn in $\Gamma_6$. Further, since $G_6$ has $O(n)$ vertices, by applying algorithm $\mathcal A$ ($\mathcal A_R$) we can construct an $O(f(n))$-step (an $O(f_R(n))$-step) upward planar morph $\mathcal M_{5,6}$ from $\Gamma_5$ to $\Gamma_6[G_5]$. Since $\Gamma'=\Gamma_5[G]$ and $\Gamma^* = \Gamma_6[G]$, the restriction of $\mathcal M_{5,6}$ to the vertices and edges of $G$ provides an $O(f(n))$-step (an $O(f_R(n))$-step) upward planar morph $\mathcal M^*$ from $\Gamma'$ to $\Gamma^*$, which proves Property~(iv). This concludes the proof of the lemma.
	%
	% \begin{figure}[tb]
	% \centering
	% \subfloat[]
	% {\includegraphics[page=2,width=.3\textwidth]{figures/gadget}\label{fig:augmentation-VbelowU}}
	% \hfil
	% \subfloat[]
	% {\includegraphics[page=1,width=.3\textwidth]{figures/gadget}\label{fig:augmentation-g1}}
	% \caption{Illustrations for the proof of \autoref{lem:morphing-between-augmentations}. 
	% The triple $[u,v,w]$ of consecutive switches of $f$ form a pocket of $f$.
	% }
	% \end{figure}
\end{proof}

We are now ready to prove the following.

\newcommand{\metaalgorithmplane}{Let $\Gamma_0$ and $\Gamma_1$ be two upward planar drawings of an $n$-vertex (reduced) upward plane graph $G$. 
	
Suppose that an algorithm $\mathcal A$ ($\mathcal A_R$) exists that constructs an $f(r)$-step (an $f_R(r)$-step) upward planar morph between any two upward planar drawings of an $r$-vertex (reduced) plane $st$-graph.

There exists an $O(n\cdot f(n))$-step (an $O(n \cdot f_R(n))$-step) upward planar morph from $\Gamma_0$ to~$\Gamma_1$.}

\begin{theorem}\label{th:meta-algorithm-plane}
\metaalgorithmplane
\end{theorem}

\begin{proof}
By \autoref{le:biconnected-reduced-augmentation}, we can assume that $G$ is biconnected.	

Denote by $\ell(G)$ the number of switches labeled \texttt{large} in the upward planar embedding of $G$. In order to prove the statement, we show that there exists a \mbox{$((2\ell(G)-3) \cdot f(n))$-step} (a $((2\ell(G)-3) \cdot  f_R(n))$-step) upward planar morph from $\Gamma_0$ to $\Gamma_1$, if $G$ is a (reduced) upward plane graph. Since $\ell(G) \in O(n)$, the statement follows. The proof is by induction on $\ell(G)$.
	
In the base case $\ell(G)=2$ and thus $G$ is a (reduced) plane $st$-graph (the two switches labeled \texttt{large} are those incident to the outer face of $G$). Hence, by applying algorithm $\mathcal A$ ($\mathcal A_R$) to $\Gamma_0$ and $\Gamma_1$, we obtain an $f(n)$-step (an $f_R(n)$-step) upward planar morph from $\Gamma_0$ to $\Gamma_1$.
	
In the inductive case $\ell(G)>2$. Then there exists a face $f$ of $G$ that is not an $st$-face. Thus, by \autoref{lem:small-small-large}, there exists a pocket $[u,v,w]$ for $f$. By \autoref{lem:morphing-between-augmentations}, we can construct upward planar drawings $\Gamma'_0$ and $\Gamma'_1$ of $G$ in which $u$ sees $w$ inside $f$ and in which $u$ lies below $w$ (assuming that a directed path exists in $f$ from $v$ to $u$, the other case being symmetric), and such that there exists an $f(n)$-step (an $f_R(n)$-step) upward planar morph $\mathcal M_{start}$ from $\Gamma_0$ to $\Gamma'_0$ and an $f(n)$-step (an $f_R(n)$-step) upward planar morph $\mathcal M_{finish}$ from $\Gamma'_1$ to $\Gamma_1$.
	
Let $G^*$ be the plane graph obtained from $G$ by splitting $f$ with a directed edge $uw$. The graph $G^*$ is an upward plane graph whose upward planar embedding is constructed by assigning to each switch in $G^*$ the same label \texttt{small} or \texttt{large} it has in $G$. Then $\ell(G^*) = \ell(G)-1$, since $u$ is not a switch in $G^*$. Further, $G^*$ is reduced if $G$ is reduced, since there exists no directed path in $G$ passing first through $u$ and then through $w$, given that $u$ is a sink of $G$.

Let $\Gamma^*_0$ and $\Gamma^*_1$ be the upward planar straight-line drawings of $G^*$ obtained by drawing the directed edge $uw$ as a straight-line segment connecting $u$ and $w$ in $\Gamma'_0$ and in $\Gamma'_1$, respectively. By the inductive hypothesis and since $V(G^*)=V(G)$, we can construct a $((2\ell(G^*)-3)\cdot f(n))$-step (a $((2\ell(G^*)-3)\cdot f_R(n))$-step) upward planar morph from $\Gamma^*_0$ to $\Gamma^*_1$. Observe that, since $G \subset G^*$, restricting each drawing in $\mathcal M^*$ to the vertices and edges of $G$ yields a $\big( (2\ell(G)-5)\cdot f(n) \big)$-step upward planar morph $\mathcal M^-$ of $G$ from $\Gamma'_0$ to $\Gamma'_1$. Therefore, by concatenating morphs $\mathcal M_{start}$, $\mathcal M^-$, and $\mathcal M_{finish}$, we obtain a $\big( (2\ell(G)-3) \cdot f(n) \big)$-step (a $\big( (2\ell(G)-3) \cdot f_R(n) \big)$-step) upward planar morph of $G$ from $\Gamma_0$ to $\Gamma_1$. This concludes the proof.
\end{proof}

\autoref{th:st-reduced-graphs}, \autoref{co:st-main}, and \autoref{th:meta-algorithm-plane} imply the following main result.

\begin{theorem}\label{th:upward-morph} 
Let $\Gamma_0$ and $\Gamma_1$ be two upward planar drawings of the same $n$-vertex (reduced) upward plane graph. There exists an $O(n^2)$-step (an $O(n)$-step) upward planar morph from $\Gamma_0$ to $\Gamma_1$.
\end{theorem}

\section{Conclusions and Open Problems}\label{se:conclusions}

In this paper, we addressed for the first time the problem of morphing upward planar straight-line drawings. We proved that an upward planar morph between any two upward planar straight-line drawings of the same upward plane graph always exists; such a morph consists of a quadratic number of linear morphing steps. The quadratic bound can be improved to linear for reduced upward plane graphs and for plane $st$-graphs, and to constant for reduced plane $st$-graphs. All our algorithms can be implemented in polynomial time in the real RAM model.

Our algorithms assume the (undirected) connectivity of the upward planar graph whose drawings have to be morphed. However, we believe that the techniques presented in~\cite{DBLP:journals/siamcomp/AlamdariABCLBFH17} in order to deal with disconnected graphs can be applied also to our setting with only minor modifications.

Several problems are left open by our research. In our opinion the most interesting question is whether an $O(1)$-step upward planar morph between any two upward planar drawings of the same maximal plane $st$-graph exists. In case of a positive answer, by \autoref{th:st-graphs-maximal} and \autoref{th:meta-algorithm-plane}, an optimal $O(n)$-step upward planar morph would exist between any two upward planar drawings of the same $n$-vertex upward plane graph. In case of a negative answer, it would be interesting to find broad classes of upward plane graphs that admit upward planar morphs with a sub-linear number of steps. In particular, we ask whether series-parallel digraphs~\cite{DBLP:journals/ijcga/BertolazziCBTT94,DBLP:journals/siamcomp/CohenBTT95} admit upward planar morphs with $O(1)$ steps.  

\bibliographystyle{abbrv}
\bibliography{bibliography}

\begin{thebibliography}{10}

\bibitem{aadddhlrssw-cpwlv-11}
O.~Aichholzer, G.~Aloupis, E.~D. Demaine, M.~L. Demaine, V.~Dujmovic,
  F.~Hurtado, A.~Lubiw, G.~Rote, A.~Schulz, D.~L. Souvaine, and A.~Winslow.
\newblock Convexifying polygons without losing visibilities.
\newblock In {\em 23rd Canadian Conference on Computational Geometry (CCCG
  '11)}, 2011.

\bibitem{DBLP:journals/siamcomp/AlamdariABCLBFH17}
S.~Alamdari, P.~Angelini, F.~Barrera{-}Cruz, T.~M. Chan, G.~{Da Lozzo}, G.~{Di
  Battista}, F.~Frati, P.~Haxell, A.~Lubiw, M.~Patrignani, V.~Roselli,
  S.~Singla, and B.~T. Wilkinson.
\newblock How to morph planar graph drawings.
\newblock {\em {SIAM} J. Comput.}, 46(2):824--852, 2017.

\bibitem{aacdfl-mpgdpns-13-c}
S.~Alamdari, P.~Angelini, T.~M. Chan, G.~{Di Battista}, F.~Frati, A.~Lubiw,
  M.~Patrignani, V.~Roselli, S.~Singla, and B.~T. Wilkinson.
\newblock Morphing planar graph drawings with a polynomial number of steps.
\newblock In S.~Khanna, editor, {\em 24th Annual ACM-SIAM Symposium on Discrete
  Algorithms (SODA '13)}, pages 1656--1667, 2013.

\bibitem{Angelini-optimal-14}
P.~Angelini, G.~{Da Lozzo}, G.~{Di Battista}, F.~Frati, M.~Patrignani, and
  V.~Roselli.
\newblock Morphing planar graph drawings optimally.
\newblock In J.~Esparza, P.~Fraigniaud, T.~Husfeldt, and E.~Koutsoupias,
  editors, {\em 41st International Colloquium on Automata, Languages and
  Programming (ICALP '14)}, volume 8572 of {\em LNCS}, pages 126--137.
  Springer, 2014.

\bibitem{DBLP:conf/compgeom/AngeliniLFLPR15}
P.~Angelini, G.~{Da Lozzo}, F.~Frati, A.~Lubiw, M.~Patrignani, and V.~Roselli.
\newblock Optimal morphs of convex drawings.
\newblock In {\em Symposium on Computational Geometry}, volume~34 of {\em
  LIPIcs}, pages 126--140. Schloss Dagstuhl - Leibniz-Zentrum fuer Informatik,
  2015.

\bibitem{afpr-mpgde-13}
P.~Angelini, F.~Frati, M.~Patrignani, and V.~Roselli.
\newblock Morphing planar graph drawings efficiently.
\newblock In S.~Wismath and A.~Wolff, editors, {\em Proc. 21st International
  Symposium on Graph Drawing (GD '13)}, LNCS. Springer, 2013.

\bibitem{Barrera-unidirectional}
F.~Barrera-Cruz, P.~Haxell, and A.~Lubiw.
\newblock Morphing planar graph drawings with unidirectional moves.
\newblock In {\em Mexican Conference on Discrete Mathematics and Computational
  Geometry}, pages 57--65, 2013.
\newblock also available at \url{http://arxiv.org/abs/1411.6185}.

\bibitem{DBLP:journals/ijcga/BertolazziCBTT94}
P.~Bertolazzi, R.~F. Cohen, G.~{Di Battista}, R.~Tamassia, and I.~G. Tollis.
\newblock How to draw a series-parallel digraph.
\newblock {\em Int. J. Comput. Geometry Appl.}, 4(4):385--402, 1994.

\bibitem{DBLP:journals/algorithmica/BertolazziBLM94}
P.~Bertolazzi, G.~{Di Battista}, G.~Liotta, and C.~Mannino.
\newblock Upward drawings of triconnected digraphs.
\newblock {\em Algorithmica}, 12(6):476--497, 1994.

\bibitem{DBLP:journals/siamcomp/BertolazziBMT98}
P.~Bertolazzi, G.~{Di Battista}, C.~Mannino, and R.~Tamassia.
\newblock Optimal upward planarity testing of single-source digraphs.
\newblock {\em {SIAM} J. Comput.}, 27(1):132--169, 1998.

\bibitem{biedl2013morphing}
T.~Biedl, A.~Lubiw, M.~Petrick, and M.~Spriggs.
\newblock Morphing orthogonal planar graph drawings.
\newblock {\em ACM Transactions on Algorithms (TALG)}, 9(4):29, 2013.

\bibitem{BLS-orth}
T.~C. Biedl, A.~Lubiw, and M.~J. Spriggs.
\newblock Morphing planar graphs while preserving edge directions.
\newblock In P.~Healy and N.~S. Nikolov, editors, {\em 13th International
  Symposium on Graph Drawing (GD '05)}, volume 3843 of {\em LNCS}, pages
  13--24. Springer, 2006.

\bibitem{c-dprc-44}
S.~S. Cairns.
\newblock Deformations of plane rectilinear complexes.
\newblock {\em American Mathematical Monthly}, 51(5):247--252, 1944.

\bibitem{DBLP:journals/siamcomp/CohenBTT95}
R.~F. Cohen, G.~{Di Battista}, R.~Tamassia, and I.~G. Tollis.
\newblock Dynamic graph drawings: Trees, series-parallel digraphs, and planar
  st-digraphs.
\newblock {\em {SIAM} J. Comput.}, 24(5):970--1001, 1995.

\bibitem{DETT}
G.~{Di Battista}, P.~Eades, R.~Tamassia, and I.~G. Tollis.
\newblock {\em Graph Drawing: Algorithms for the Visualization of Graphs}.
\newblock Prentice-Hall, 1999.

\bibitem{DBLP:journals/tcs/BattistaT88}
G.~{Di Battista} and R.~Tamassia.
\newblock Algorithms for plane representations of acyclic digraphs.
\newblock {\em Theor. Comput. Sci.}, 61:175--198, 1988.

\bibitem{DiBattista:1992:ARS:149153.149159}
G.~Di~Battista, R.~Tamassia, and I.~G. Tollis.
\newblock Area requirement and symmetry display of planar upward drawings.
\newblock {\em Discrete Comput. Geom.}, 7(4):381--401, Apr. 1992.

\bibitem{Fary}
I.~F\'ary.
\newblock On straight-line representation of planar graphs.
\newblock {\em Acta Scientiarum Mathematicarum (Szeged)}, 11:229--233, 1948.

\bibitem{fgw-nupo-12}
F.~Frati, J.~Gudmundsson, and E.~Welzl.
\newblock On the number of upward planar orientations of maximal planar graphs.
\newblock In K.~Chao, T.~Hsu, and D.~Lee, editors, {\em 23rd International
  Symposium on Algorithms and Computation ({ISAAC} '12)}, volume 7676 of {\em
  Lecture Notes in Computer Science}, pages 413--422. Springer, 2012.

\bibitem{Garg:2002:CCU:586839.586865}
A.~Garg and R.~Tamassia.
\newblock On the computational complexity of upward and rectilinear planarity
  testing.
\newblock {\em SIAM J. Comput.}, 31(2):601--625, Feb. 2002.

\bibitem{DBLP:journals/jda/HongN10}
S.~Hong and H.~Nagamochi.
\newblock Convex drawings of hierarchical planar graphs and clustered planar
  graphs.
\newblock {\em J. Discrete Algorithms}, 8(3):282--295, 2010.

\bibitem{Hopcroft:1974:EPT:321850.321852}
J.~Hopcroft and R.~Tarjan.
\newblock Efficient planarity testing.
\newblock {\em J. ACM}, 21(4):549--568, Oct. 1974.

\bibitem{Mel}
K.~Mehlhorn.
\newblock {\em Data Structures and Algorithms: Multi-dimensional Searching and
  Computational Geometry}, volume~3.
\newblock Springer, 1984.

\bibitem{r-mvdg-14}
V.~Roselli.
\newblock {\em Morphing and Visiting Drawings of Graphs}.
\newblock PhD thesis, Universit\`a degli Studi di Roma ``Roma Tre'', Dottorato
  di Ricerca in Ingegneria, Sezione Informatica ed Automazione, XXVI Ciclo,
  2014.

\bibitem{RT-rpl-86}
P.~Rosenstiehl and R.~E. Tarjan.
\newblock Rectilinear planar layouts and bipolar orientations of planar graphs.
\newblock {\em Discrete {\&} Computational Geometry}, 1:343--353, 1986.

\bibitem{Schnyder:1990:EPG:320176.320191}
W.~Schnyder.
\newblock Embedding planar graphs on the grid.
\newblock In {\em Proceedings of the First Annual ACM-SIAM Symposium on
  Discrete Algorithms}, SODA '90, pages 138--148, Philadelphia, PA, USA, 1990.
  Society for Industrial and Applied Mathematics.

\bibitem{s-ocrsui-17}
H.~L. Smith.
\newblock On continuous representations of a square upon itself.
\newblock {\em Annals of Mathematics}, 19(2):137--141, 1917.

\bibitem{TT86-avrpg-86}
R.~Tamassia and I.~G. Tollis.
\newblock Algorithms for visibility representations of planar graphs.
\newblock In B.~Monien and G.~Vidal{-}Naquet, editors, {\em {STACS} 86, 3rd
  Annual Symposium on Theoretical Aspects of Computer Science, Orsay, France,
  January 16-18, 1986, Proceedings}, volume 210 of {\em Lecture Notes in
  Computer Science}, pages 130--141. Springer, 1986.

\bibitem{TT-uavr-86}
R.~Tamassia and I.~G. Tollis.
\newblock A unified approach a visibility representation of planar graphs.
\newblock {\em Discrete {\&} Computational Geometry}, 1:321--341, 1986.

\bibitem{t-dpg-83}
C.~Thomassen.
\newblock Deformations of plane graphs.
\newblock {\em Journal of Combinatorial Theory, Series B}, 34(3):244--257,
  1983.

\bibitem{tietze-usaeqass-14}
H.~Tietze.
\newblock \"{U}ber stetige abbildungen einer quadratfl\"{a}che auf sich selbst.
\newblock {\em Rendiconti del Circolo Matematico di Palermo}, 38(1):247--304,
  1914.

\bibitem{DBLP:conf/compgeom/GoethemV18}
A.~van Goethem and K.~Verbeek.
\newblock Optimal morphs of planar orthogonal drawings.
\newblock In {\em 34th International Symposium on Computational Geometry, SoCG
  2018, June 11-14, 2018, Budapest, Hungary}, pages 42:1--42:14, 2018.

\end{thebibliography}

\end{document}